\documentclass[a4paper,UKenglish,autoref, thm-restate]{lipics-v2019}

\title{Diverse Collections in Matroids and Graphs}

\author{Fedor V. Fomin}{University of Bergen, Norway \and \url{http://www.ii.uib.no/~fomin/}}{fomin@ii.uib.no}{https://orcid.org/0000-0003-1955-4612}{Supported by the Research Council of Norway via the project ``MULTIVAL'' (grant no. 263317).}
\author{Petr A. Golovach}{University of Bergen, Norway \and \url{https://folk.uib.no/pgo041/}}{Petr.Golovach@uib.no}{https://orcid.org/0000-0002-2619-2990}{Supported by the Research Council of Norway via the project ``MULTIVAL'' (grant no. 263317).}
\author{Fahad Panolan}{Department of Computer Science and Engineering, IIT Hyderabad, India \and \url{https://iith.ac.in/~fahad/}}{fahad@cse.iith.ac.in}{https://orcid.org/0000-0001-6213-8687}{Seed grant, IIT Hyderabad (SG/IITH/F224/2020-21/SG-79)}
\author{Geevarghese Philip}{Chennai Mathematical Institute, India \and UMI ReLaX
  \and \url{https://www.cmi.ac.in/~gphilip}}{gphilip@cmi.ac.in}{http://orcid.org/0000-0003-0717-7303}{}
\author{Saket Saurabh}{Institute of Mathematical Sciences, India \and University
  of Bergen, Norway \and \url{https://www.imsc.res.in/~saket}}{saket@imsc.res.in}{}{European Research Council (ERC) under the European Union’s Horizon 2020 research and innovation programme (grant no. 819416), and Swarnajayanti Fellowship
grant DST/SJF/MSA-01/2017-18.\begin{minipage}{0.1\textwidth}
    \begin{center}
        \includegraphics[scale=0.5]{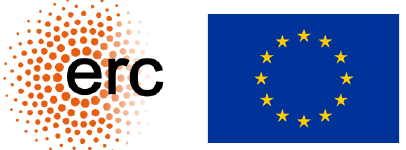}
    \end{center}
\end{minipage}}

\authorrunning{F.\,V. Fomin and P.\,A. Golovach and F. Panolan and G. Philip and
S. Saurabh} 

\Copyright{Fedor V. Fomin and Petr A. Golovach and Fahad Panolan and Geevarghese
Philip and Saket Saurabh} 

\ccsdesc[100]{\textcolor{red}{Theory of computation~Fixed parameter tractability}}  
\keywords{Matroids, Matching, Diverse solutions, Fixed-parameter tractable algorithms} 

\relatedversion{An extended abstract has been accepted for publication at STACS 2021.}

\nolinenumbers 

\hideLIPIcs  

\EventEditors{John Q. Open and Joan R. Access}
\EventNoEds{2}
\EventLongTitle{42nd Conference on Very Important Topics (CVIT 2016)}
\EventShortTitle{CVIT 2016}
\EventAcronym{CVIT}
\EventYear{2016}
\EventDate{December 24--27, 2016}
\EventLocation{Little Whinging, United Kingdom}
\EventLogo{}
\SeriesVolume{42}
\ArticleNo{23}

\usepackage{xspace}

\newcommand{\yes}{\textsc{Yes}\xspace}
\newcommand{\no}{\textsc{No}\xspace}
\newcommand{\yesinstance}{\yes-instance\xspace}
\newcommand{\noinstance}{\no-instance\xspace}
\newcommand{\dl}[1]{\textcolor{lipicsGray}{\sffamily\bfseries\mathversion{bold}#1}}

\DeclareMathOperator{\operatorClassFPT}{{\sf FPT}\xspace}
\newcommand{\classFPT}{\ensuremath{\operatorClassFPT}\xspace}
\DeclareMathOperator{\operatorClassP}{{\sf P}}
\newcommand{\classP}{\ensuremath{\operatorClassP}\xspace}
\DeclareMathOperator{\operatorClassNP}{{\sf NP}}
\newcommand{\classNP}{\ensuremath{\operatorClassNP}\xspace}
\newcommand{\Oh}{\ensuremath{\mathcal{O}}}

\newcommand{\naturals}{\ensuremath{\mathbb{N}}\xspace}
\newcommand{\gfq}{\textsf{GF(}q{\textsf{)}}\xspace}
\newcommand{\bfA}{\textbf{A}\xspace}

\newcommand{\rank}{\textsf{rank}}
\newcommand{\corank}{\textsf{corank}}
\newcommand{\cl}{\textsf{cl}}
\newcommand{\MStar}{\ensuremath{M^{*}}\xspace}
\newcommand{\prob}{{\operatorname{Pr}}}

\newtheorem{observation}{Observation}

\newcommand{\sd}{\bigtriangleup}
\newcommand{\pname}{\textsc}
\newcommand{\ProblemFormat}[1]{\pname{#1}}
\newcommand{\ProblemIndex}[1]{\index{problem!\ProblemFormat{#1}}}
\newcommand{\ProblemName}[1]{\ProblemFormat{#1}\ProblemIndex{#1}{}\xspace}

\newcommand{\probWDB}{\ProblemName{Weighted Diverse Bases}}
\newcommand{\probPMD}{\ProblemName{Diverse Perfect Matchings}}
\newcommand{\probWDCIS}{\ProblemName{Weighted Diverse Common Independent Sets}}
\newcommand{\probWMI}{\ProblemName{Weighted Matroid Intersection}}
\newcommand{\probMI}{\ProblemName{Matroid Intersection}}

\usepackage{framed}
\usepackage{xifthen}
\usepackage{tikz}
\usetikzlibrary{calc}
\newlength{\RoundedBoxWidth}
\newsavebox{\GrayRoundedBox}
\newenvironment{GrayBox}[1]%
   {\setlength{\RoundedBoxWidth}{.93\textwidth}
    \def\boxheading{#1}
    \begin{lrbox}{\GrayRoundedBox}
       \begin{minipage}{\RoundedBoxWidth}}%
   {   \end{minipage}
    \end{lrbox}
    \begin{center}
    \begin{tikzpicture}%
       \node(Text)[draw=black!20,fill=white,rounded corners,%
             inner sep=2ex,text width=\RoundedBoxWidth]%
             {\usebox{\GrayRoundedBox}};
        \coordinate(x) at (current bounding box.north west);
        \node [draw=white,rectangle,inner sep=3pt,anchor=north west,fill=white] 
        at ($(x)+(6pt,.75em)$) {\boxheading};
    \end{tikzpicture}
    \end{center}}     
\newenvironment{defproblemx}[2][]{\noindent\ignorespaces%
                                \FrameSep=6pt%
                                \parindent=0pt%
                \vspace*{-1.5em}
                \ifthenelse{\isempty{#1}}{%
                  \begin{GrayBox}{\textsc{#2}}%
                }{%
                  \begin{GrayBox}{\textsc{#2}  parameterized by~{#1}}%
                }
                \begin{tabular*}{\textwidth}{@{\hspace{.1em}} >{\itshape} p{1.8cm} p{0.8\textwidth} @{}}%
            }{
                \end{tabular*}%
                \end{GrayBox}%
                \ignorespacesafterend
            }  

\newcommand{\defproblema}[3]{
  \begin{defproblemx}{#1}
    Input:  & #2 \\
    Task: & #3
  \end{defproblemx}
}%

\begin{document}

\maketitle

\begin{abstract}
  We investigate the parameterized complexity of finding diverse sets of
  solutions to three fundamental combinatorial problems, two from the theory of
  matroids and the third from graph theory. The input to the \probWDB problem
  consists of a matroid \(M\), a weight function \(\omega:E(M)\to\naturals\),
  and integers \(k\geq 1, d\geq 0\). The task is to decide if there is a
  collection of \(k\) \emph{bases} \(B_{1}, \dotsc, B_{k}\) of \(M\) such that
  the weight of the symmetric difference of any pair of these bases is at least
  \(d\). This is a diverse variant of the classical matroid base packing
  problem. The input to the \probWDCIS problem consists of two matroids
  \(M_{1},M_{2}\) defined on the same ground set \(E\), a weight function
  \(\omega:E\to\naturals\), and integers \(k\geq 1, d\geq 0\). The task is to
  decide if there is a collection of \(k\) \emph{common independent sets}
  \(I_{1}, \dotsc, I_{k}\) of \(M_{1}\) and \(M_{2}\) such that the weight of
  the symmetric difference of any pair of these sets is at least \(d\). This is
  motivated by the classical weighted matroid intersection problem. The input to
  the \probPMD problem consists of a graph \(G\) and integers \(k\geq 1, d\geq
  0\). The task is to decide if \(G\) contains \(k\) \emph{perfect matchings}
  \(M_{1},\dotsc,M_{k}\) such that the symmetric difference of any two of these
  matchings is at least \(d\).

  The underlying problem of finding \emph{one} solution (basis, common
  independent set, or perfect matching) is known to be doable in polynomial time
  for each of these problems, and \probPMD is known to be \classNP-hard for
  \(k=2\). We show that \probWDB and \probWDCIS are both \classNP-hard. We show
  also that \probPMD cannot be solved in polynomial time (unless
  \(\classP=\classNP\)) even for the case \(d=1\). We derive fixed-parameter
  tractable (\classFPT) algorithms for all three problems with \((k,d)\) as the
  parameter.

  The above results on matroids are derived under the assumption that the input
  matroids are given as \emph{independence oracles}. For \probWDB we present a
  polynomial-time algorithm that takes a representation of the input matroid
  over a finite field and computes a \(poly(k,d)\)-sized kernel for the problem.
\end{abstract}


\section{Introduction}\label{sec:intro}

In this work we study the parameterized complexity of finding \emph{diverse
  collections of solutions} to three basic algorithmic problems. Two of these
problems arise in the theory of matroids. The third problem belongs to the
domain of graph theory, and its restriction to bipartite graphs can be rephrased
as a question about matroids. Each of these is a fundamental algorithmic problem
in its respective domain.
\paragraph*{Diverse \classFPT Algorithms.}
Nearly every existing approach to solving algorithmic
problems focuses on finding \emph{one solution of good quality} for a given
input. 
For algorithmic problems which are---eventually---motivated by problems from the
real world, finding ``one good solution'' may not be of much use for
practitioners of the real-world discipline from which the problem was originally
drawn. This is primarily because the process of abstracting out a ``nice''
algorithmic problem from a ``messy'' real-world problem invariably involves
throwing out a lot of ``side information'' which is very relevant to the
real-world problem, but is inconvenient, difficult, or even impossible to model
mathematically.

The other extreme of enumerating \emph{all} (or even all minimal or maximal)
solutions to an input instance is also usually not a viable solution. 
A third approach is to look for \emph{a few solutions of good quality} which are
``far away'' from one another according to an appropriate notion of distance.
The intuition is that given such a collection of ``diverse'' solutions, an
end-user can choose one of the solutions by factoring in the ``side
information'' which is absent from the algorithmic
model. 

These and other considerations led Fellows to propose \emph{the Diverse \(X\)
  Paradigm}~\cite{fellows2018diverseXparadigm}. Here ``\(X\)'' is a placeholder
for an optimization problem, and the goal is to study the fixed-parameter
tractability of finding a diverse collection of good-quality solutions for
\(X\). 
Recall that the \emph{Hamming distance} of two sets is the size of their
symmetric difference. A natural measure of diversity for problems whose
solutions are subsets of some kind is the \emph{minimum} Hamming distance of any
pair of solutions. In this work we study the parameterized complexity of finding
diverse collections of solutions for three fundamental problems with this
diversity measure and its weighted variant.

\paragraph*{Our problems.} Let \(M\) be a matroid on ground set \(E(M)\) and
with rank function \(r()\). The departure point of our work is the classical
theorem of Edmonds from 1965~\cite{edmonds1965lehman} about matroid partition.
This theorem states that a matroid $M$ has $k$ \emph{pairwise disjoint} bases if
and only if, for every subset $X$ of $E(M)$,
\[
k \cdot r(X) +|E(M) -X|\geq k \cdot r(M).
\]
An important algorithmic consequence of this result is that given access to an
independence oracle for a matroid \(M\), one can find a maximum number of
\emph{pairwise disjoint bases} of \(M\) in polynomial time (See, e.g.,
\cite[Theorem~42.5]{schrijver2003combinatorial}). This in turn implies, for
instance, that the maximum number of pairwise edge-disjoint spanning trees of a
connected graph can be found in polynomial time.

We take a fresh look at this fundamental result of Edmonds: what happens if we
don't insist that the bases be pairwise disjoint, and instead allow them to have
some pairwise intersection? We work in the weighted setting where each element
\(e\) of the ground set \(E(M)\) has a positive integral weight \(\omega(e)\)
associated with it, and the weight of a subset \(X\) of \(E(M)\) is the sum of
the weights of the elements in \(X\). The relaxed version of the pairwise
disjoint bases problem is then: Given an independence oracle for a matroid \(M\)
and integers \(k,d\) as input, find if \(M\) has \(k\) bases \(B_{1}, \dotsc,
B_{k}\) such that for every pair of bases \(B_{i}, B_{j}\;;\;i\neq j\) the
weight \(\omega(B_{i} \sd B_{j})\) of their symmetric difference is at least
\(d\). We call this the \probWDB problem:

\defproblema{\probWDB}%
{A matroid $M$, a weight function $\omega\colon E(M)\rightarrow \mathbb{N}$, and
  integers $k\geq 1$ and $d\geq 0$.}%
{Decide whether there are bases $B_1,\ldots,B_k$ of $M$ such that $\omega(B_i
  \sd B_j)\geq d$ holds for all distinct $i,j\in\{1,\ldots,k\}$. } Due to the
expressive power of matroids \probWDB captures many interesting computational
problems. We list a few examples; in each case the weight function assigns
positive integral weights, \(k \geq 1\) and \(d \geq 0\) are integers, and we
say that a collection of objects is \emph{diverse} if the weight of the
symmetric difference of each pair of objects in the collection is at least
\(d\). When $M$ is a graphic matroid \probWDB corresponds to finding diverse
\emph{spanning trees} in an edge-weighted graph. When $M$ is a vector matroid
then this is the problem of finding diverse \emph{column (or row) bases} of a
matrix with column (or row) weights. And when \(M\) is a transversal matroid on
a weighted ground set then this problem corresponds to finding diverse
\emph{systems of distinct representatives}.

Another celebrated result of Edmonds is the \emph{Matroid Intersection
  Theorem}~\cite{Edmonds70} which states that if \(M_{1}, M_{2}\)
are matroids on a common ground set \(E\) and with rank functions \(r_{1},
r_{2}\), respectively, then the size of a largest subset of \(E\) which is
independent in both \(M_{1}\) and \(M_{2}\) (a \emph{common independent set}) is
given by

\[
\min_{T \subseteq E}(r_{1}(T) + r_{2}(E-T)).
\]

Edmonds showed that given access to independence oracles for \(M_{1}\) and
\(M_{2}\), a maximum-size common independent set of \(M_{1}\) and \(M_{2}\) can
be found in polynomial time~\cite{Edmonds70}. This is called the \probMI
problem. Frank~\cite{Frank81} found a polynomial-time algorithm for the more
general \probWMI problem where the input has an additional weight function
\(\omega:E\to\naturals\) and the goal is to find a common independent set of the
maximum \emph{weight}. The second problem that we address in this work is a
``diverse'' take on \probWMI where we replace the maximality requirement on
individual sets with a lower bound on the weight of their symmetric difference.
Given \(M_{1}, M_{2},\omega\) as above and integers \(k,d\), we ask if 
there are \(k\) common independent sets whose pairwise symmetric differences
have weight at least \(d\) each; this is the \probWDCIS problem.

\defproblema{\probWDCIS}%
{Matroids $M_1$ and $M_2$ with a common ground set $E$, a weight function
  $\omega\colon E\rightarrow \mathbb{N}$, and integers $k\geq 1$ and $d\geq
  0$.}%
{Decide whether there are sets $I_1,\ldots,I_k\subseteq E$ such that $I_i$ is
  independent in both $M_1$ and $M_2$ for every $i\in\{1,\ldots,k\}$ and
  $\omega(I_i \sd I_j)\geq d$ for all distinct $i,j\in\{1,\ldots,k\}$. }

\probWDCIS also captures many interesting algorithmic problems. We give a few
examples (\emph{cf.}~\cite[Section~41.1a]{schrijver2003combinatorial}). We use
``diverse'' here in the sense defined above. Given a bipartite graph \(G\) with
edge weights, \probWDCIS can be used to ask if there is a diverse collection of
\(k\) \emph{matchings} in \(G\). A \emph{partial orientation} of an undirected
graph \(G\) is a directed graph obtained by (i) assigning directions to some
subset of edges of \(G\) and (ii) deleting the remaining edges. Given an
undirected graph \(G = (V,E)\) with edge weights and a function \(\iota: V \to
\mathbb{N}\), we say that a partial orientation \(\mathcal{O}\) of \(G\)
\emph{respects} \(\iota\) if the in-degree of every vertex \(v\) in
\(\mathcal{O}\) is at most \(\iota(v)\). We can use \probWDCIS to ask if there
is a diverse collection of \(k\) partial orientations of \(G\), all of which
respect \(\iota\). For a third example, let \(G = (V, E)\) be an undirected
graph with edge weights, in which each edge is assigned a---not necessarily
distinct---\emph{color}. A \emph{colorful forest} in \(G\) is any subgraph of
\(G\) which is a forest in which no two edges have the same color. We can use
\probWDCIS to ask if there is a diverse collection of \(k\) colorful forests in
\(G\).

Finding whether a bipartite graph has a perfect matching or not is a well-known
application of \probMI (\cite[Section~41.1a]{schrijver2003combinatorial}). The
third problem that we study in this work is a diverse version of the former
problem, extended to general graphs. Note that there is no known interpretation
of the problem of finding perfect matchings in (general) undirected graphs in
terms of \probMI. 

\defproblema{\probPMD}%
{An undirected graph $G$ on \(n\) vertices, and integers $k\geq 1$ and $d\geq
  0$.}%
{Decide whether there are perfect matchings $M_1,\ldots,M_k$ of $G$ such that
  $|M_i \sd M_j|\geq d$ for all distinct $i,j\in\{1,\ldots,k\}$. }

\paragraph*{Our results.} We assume throughout that matroids in the input are
given in terms of an \emph{independence oracle}. Recall that with this
assumption, we can find \emph{one} basis of the largest weight and \emph{one}
common independent set (of two matroids) of the largest weight, both in
polynomial time. In contrast, we show that the diverse versions \probWDB and
\probWDCIS are both \classNP-hard, even when the weights are expressed in
unary\footnote{See \autoref{thm:WDB_NP_hard} for an alternative hardness result
  for \probWDB.}.

\begin{restatable}{theorem}{nphardness}
  \label{thm:np-hard}
  Both \probWDB and \probWDCIS are strongly \classNP-complete, even on the
  uniform matroids $U_n^3$.
\end{restatable}

Given this hardness, we analyze the parameterized complexity of these problems
with \(d,k\) as the parameters. Our first result is that \probWDB is
fixed-parameter tractable (\classFPT) under this parameterization:

\begin{restatable}{theorem}{divBasesFPT}
  \label{thm:divBasesFPT}
\probWDB can be solved in $2^{\Oh(dk^2(\log k+\log d))}\cdot |E(M)|^{\Oh(1)}$ time.
\end{restatable}

We have a stronger result if the input matroid is given as a representation over
a finite field (and not just as a ``black box'' independence oracle): in this
case we show that \probWDB admits a \emph{polynomial kernel} with this
parameterization.

\begin{restatable}{theorem}{kernel}
  \label{thm:kern}
  Given a representation of the matroid $M$ over a finite field \gfq as input,
  we can compute a kernel of \probWDB of size $\Oh(k^6d^4\log q)$.
\end{restatable}

We then show that our second matroid-related diverse problem is also \classFPT
under the same parameterization.

\begin{restatable}{theorem}{divComIndSetsFPT}
\label{thm:FPT-WDCIS}
\probWDCIS can be solved in $2^{\Oh(k^3d^2\log(kd))}\cdot |E|^{\Oh(1)}$ time. 
\end{restatable}

We now turn to the problem of finding diverse perfect matchings. \probPMD is
known to be \classNP-hard already when \(k=2\) and \(G\) is a \(3\)-regular
graph~\cite{holyer1981np,DBLP:journals/corr/abs-2009-04567}. Since all perfect
matchings of a graph have the same size the symmetric difference of two distinct
perfect matchings is at least \(2\). Setting \(d=1\) in \probPMD is thus
equivalent to asking whether \(G\) has at least \(k\) distinct perfect
matchings. Since a bipartite graph on \(n\) vertices has at most
\(\frac{n}{2}!\) perfect matchings and since \(\log(\frac{n}{2}!) = \Oh(n \log
n)\) we get---using binary search---that there is a polynomial-time Turing
reduction from the problem of \emph{counting} the number of perfect matchings in
a bipartite graph to \probPMD instances with \(d=1\). Since the former problem
is \textsf{\#P}-complete~\cite{valiant1979complexity} we get

\begin{theorem}
  \label{thm:DPM_NP_hard}
  \probPMD with \(d=1\) cannot be solved in time polynomial in \(n=|V(G)|\) even
  when graph \(G\) is bipartite, unless \(\classP = \classNP\).
\end{theorem}

Thus we get that \probPMD is unlikely to have a polynomial-time algorithm even
if \emph{one} of the two numbers \(k,d\) is a small constant. We show that the
problem \emph{does} have a (randomized) polynomial-time algorithm when
\emph{both} these parameters are bounded; \probPMD is (randomized) \classFPT
with \(k\) and \(d\) as parameters:
\begin{restatable}{theorem}{divMatchingFPT}
\label{thm:matching-fpt}
There is an algorithm that given an instance of \probPMD, runs in time
$2^{2^{\Oh(kd)}}n^{\Oh(1)}$ and outputs the following: If the input is a
\noinstance then the algorithm outputs \no. Otherwise the algorithm outputs \yes
with probability at least $1-\frac{1}{e}$.
\end{restatable}

Note that \autoref{thm:matching-fpt} implies, in particular, that \probPMD can
be solved in (randomized) polynomial time when \(kd \leq c_{1} + \frac{\log\log
  n}{c_{2}}\) holds for some constants \(c_{1}, c_{2}\) which depend on the
constant hidden by the \(\Oh()\) notation.

\paragraph*{Our methods.} We prove the \classNP-hardness results
(\autoref{thm:np-hard}) by reduction from the \textsc{$3$-Partition} problem. To
show that \probWDB is \classFPT (\autoref{thm:divBasesFPT}) we observe first
that if the input matroid \(M\) contains a set of size \(\Omega(kd)\) which is
\emph{both} independent \emph{and} co-independent in \(M\) then the input is a
\yes instance of \probWDB (\autoref{lem:ind-coind}). We can check for the
existence of such a set in time polynomial in \(|E(M)|\), so we assume without
loss of generality that no such set exists. We then show that starting with an
arbitrary basis of \(M\) and repeatedly applying the greedy algorithm
(\autoref{prop:greedy}) \(poly(k,d)\)-many times we can find, in time polynomial
in \((|E(M)| + k + d)\), (i) a subset \(S^{*} \subseteq E(M)\) of size
\(poly(k,d)\) and (ii) a matroid \(\widetilde{M}\) on the ground set \(S^{*}\)
such that $(\widetilde{M},\omega,k,d)$ is \emph{equivalent} to the input
instance $(M,\omega,k,d)$ (\autoref{lem:main}). We also show how to compute a
useful partition of \(E(\widetilde{M}) = S^{*}\) which speeds up the subsequent
\classFPT-time search for a diverse set of bases in \(\widetilde{M}\). The
kernelization result for \probWDB (\autoref{thm:kern}) follows directly from
\autoref{lem:main}. This ``compression lemma'' is thus the main technical
component of our algorithms for \probWDB.

To show that \probWDCIS is \classFPT (\autoref{thm:FPT-WDCIS}) we observe first
that if the two input matroids \(M_{1},M_{2}\) have a \emph{common} independent
set of size \(\Omega(kd)\) then the input is a \yes instance of \probWDCIS
(\autoref{lem:big-is}). So we assume that this is not the case, and then show
(\autoref{lem:branching}) that we can construct, in \(f(k,d)\) time, a
collection $\mathcal{F}$ of common independent sets of $M_1$ and $M_2$ of size
\(g(k,d)\) such that if the input is a \yesinstance then it has a solution
$I_1,\ldots,I_k$ with $I_i\in\mathcal{F}$ for $i\in\{1,\ldots,k\}$. The
\classFPT algorithm for \probWDCIS follows by a simple search in the collection
\(\mathcal{F}\).

Our algorithm for \probPMD\ is based on two procedures. 
\begin{description}
\item[P1] Given an undirected graph $G$ on \(n\) vertices, perfect matchings
  $M_1,\ldots, M_r$ of \(G\), and a non-negative integer $s$ as input, this
  procedure (\autoref{lem:PMDstep1}) runs in time $2^{\Oh(rs)} n^{\Oh(1)}$ and
  outputs a perfect matching $M$ of \(G\) such that $\vert M \sd M_i\vert \geq
  2s$ holds for all $i\in \{1,\ldots,r\}$ (if such a matching exists), with
  probability at least $\frac{2}{3}e^{-rs}$.
\item[P2] Given an undirected graph $G$ on \(n\) vertices, a perfect matching
  $M$ of \(G\), and non-negative integers $r,d,s$, this
  procedure(\autoref{lem:PMDstep2}) runs in time $2^{\Oh(r^2s)} n^{\Oh(1)}$, and
  outputs $r$ perfect matchings $M_1^{\star},\ldots,M_r^{\star}$ of \(G\) such
  that $\vert M \sd M_i^{\star}\vert \leq s$ holds for all $i\in \{1,\ldots,r\}$
  and $\vert M_i^{\star} \sd M_j^{\star} \vert\geq d$ holds for all distinct
  $i,j\in [r]$ (if such matchings exist), with probability at least $e^{-rs}$.
  If no such perfect matchings exist, then the algorithm outputs \no.
\end{description}

Let $(G,k,d)$ be the input instance of \probPMD. We use procedure \dl{P1} to
greedily compute a collection of matchings which are ``far apart'': We start
with an arbitrary perfect matching $M_1$. In step $i$, we have a collection of
perfect matchings $M_1,\ldots,M_{i-1}$ such that $\vert M_{j}\sd M_{j'}\vert
\geq 2^{k-i}d$ holds for any two distinct $j,j'\in \{1,\ldots,i-1\}$. We now run
procedure \dl{P1} with $r=i-1$ and $s=2^{k-i}d$ to find---if it exists---a
matching $M_{i}$ such that $\vert M_{i}\sd M_j\vert \geq 2^{k-i+1}d$ holds for
all $j\in \{1,\ldots,i\}$. By exhaustively applying \dl{P1} we get a collection
of perfect matchings $M_1,\ldots,M_q$ such that

\begin{alphaenumerate}
\item for any two distinct integers $i,j\in \{1,\ldots,q\}$, $|M_i\sd M_j|\geq
  2^{k-q+1}d$, and
\item for any other perfect matching $M\notin \{M_1,\ldots,M_q\}$, $\vert M\sd
  M_j\vert \leq 2^{k-q}d$.
\end{alphaenumerate}

Thus, if $k\leq q$, then clearly $\{M_1,\ldots,M_k\}$ is a solution. Otherwise,
let ${\cal M}=\{M_1^{\star},\ldots,M_{k}^{\star}\}$ be a hypothetical solution.
Then for each $M_{i}^{\star}$ there is a \emph{unique} matching $M_j$ in
$\{M_{1},\ldots,M_q\}$ such that $\vert M_j\sd M_{i}^{\star}\vert < 2^{(k-q)}d$
holds (\autoref{clm:uniquemat}). For each $i\in \{1,\ldots,q\}$ we guess the
number $r_i$ of perfect matchings from ${\cal M}$ that are \emph{close} to
$M_i$, and use procedure \dl{P2} to compute a set of $r_i$ diverse perfect
matchings that are close to $M_i$. The union of all the matchings computed for
all $i\in \{1,\ldots,q\}$ form a solution.

We use algebraic methods and color coding to design procedure \dl{P1}. The Tutte
matrix \bfA of an undirected graph $G$ over the field ${\mathbb F}_2[X]$ is
defined as follows, where ${\mathbb F}_2$ is the Galois field on $\{0,1\}$ and
$X=\{x_e ~\colon~ e\in E(G)\}$. The rows and columns of \bfA are labeled with
$V(G)$ and for each $e=\{u,v\}\in E(G)$, $\bfA[u,v]=\bfA[v,u]=x_e$. All other
entries in \bfA are zeros. There is a bijective correspondence between the set
of monomials of $det(\bfA)$ and the set of perfect matchings of $G$. Procedure
\dl{P1} extracts the required matching from $det(\bfA)$ using color coding.
Procedure \dl{P2} is realized using color coding and dynamic programming.
\paragraph*{Related work.} 
Recall that all bases of a matroid have the same size, and that the number of
bases of a matroid on ground set \(E\) is at most \(2^{|E|}\). So using the same
argument as for \autoref{thm:DPM_NP_hard} we get that \probWDB generalizes---via
Turing reductions---the problem of \emph{counting the number of bases} of a
matroid. Each of these reduced \probWDB instances will have \(d=1\), and a
weight function which assigns the weight \(1\) to each element in the ground
set. Counting the number of bases of a matroid is known to be
\textsf{\#P}-complete even for restricted classes of matroids such as
transversal~\cite{Colbourn95}, bircircular~\cite{Gimenez06}, and binary
matroids~\cite{Vertigan98}. Hence we have the following
alternative\footnote{Compare with \autoref{thm:np-hard}.} hardness result for
\probWDB

\begin{theorem}
  \label{thm:WDB_NP_hard}
  \probWDB cannot be solved in time polynomial in \(|E(M)|\) unless \(\classP =
  \classNP\), even when \(d=1\) and every element of the ground set \(E(M)\) has
  weight \(1\).
\end{theorem}

The study of the parameterized complexity of finding diverse sets of solutions
is a very recent development, and only a handful of results are currently known.
In the work which introduced this notion Baste et al.~\cite{BasteEtAl2019}
showed that diverse variants of a large class of graph problems which are
\classFPT when parameterized by the \emph{treewidth} of the input graph, are
also \classFPT when parameterized by the treewidth and the number of solutions
in the collection. In a second article~\cite{BasteEtAl2019b} the authors show
that for each fixed positive integer \(d\), two diverse variants---one with the
\emph{minimum} Hamming distance of any pair of solutions, and the other with the
\emph{sum} of all pairwise Hamming distances of solutions---of the
\(d\)-\textsc{Hitting Set} problem are \classFPT when parameterized by the size
of the hitting set and the number of solutions. In a recent manuscript on
diverse \classFPT algorithms~\cite{DBLP:journals/corr/abs-2009-04567} the
authors show that the problem of finding \emph{two} maximum-sized matchings in
an undirected graph such that their symmetric difference is at least \(d\), is
\classFPT when parameterized by \(d\). Note that our result on \probPMD
generalizes this to \(k \geq 2\) matchings, \emph{provided} the input graph has
a perfect matching.

In a very recent manuscript Hanaka et al.~\cite{hanaka2020finding} propose a
number of results about finding diverse solutions. We briefly summarize their
results which are germane to our work. For a collection of sets \(X_{1}, \dotsc,
X_{k}\) let \(d_{sum}(X_{1}, \dotsc, X_{k})\) denote the sum of all pairwise
Hamming distances of these sets and let \(d_{min}(X_{1}, \dotsc, X_{k})\) denote
the smallest Hamming distance of any pair of sets in the collection. Hanaka et
al. show that there is an algorithm which takes an independence oracle for a
matroid \(M\) and an integer \(k\) as input, runs in time polynomial in
\((|E(M)| + k)\), and finds a collection \(B_{1}, B_{2}, \dotsc, B_{k}\) of
\(k\) bases of \(M\) which maximizes \(d_{sum}(B_{1}, B_{2}, \dotsc, B_{k})\).
This result differs from our work on \probWDB in two key aspects. They deal with
the unweighted (counting) case, and their diversity measure is the \emph{sum} of
the pairwise symmetric differences, whereas we look at the \emph{minimum}
(weight of the) symmetric difference. These two measures are, in general, not
comparable.

Hanaka et al. also look at the complexity of finding \(k\) matchings \(M_{1},
\dotsc, M_{k}\) in a graph \(G\) where each \(M_{i}\) is of size \(t\). They
show that such collections of matchings maximizing \(d_{min}(M_{1}, \dotsc,
M_{k})\) and \(d_{sum}(M_{1}, \dotsc, M_{k})\) can be found in time
\(2^{\Oh(kt\log(kt))}\cdot |V(G)|^{\Oh(1)}\). The key difference with our work
is that their algorithm looks for matchings of a specified size \(t\) whereas
ours looks for perfect matchings, of size \(t = \frac{|V(G)|}{2}\); note that
this \(t\) does not appear in the exponential part of the running time of our
algorithm (\autoref{thm:matching-fpt}). The manuscript~\cite{hanaka2020finding}
has a variety of other interesting results on diverse \classFPT algorithms as
well.

\paragraph*{Organization of the rest of the paper.} In the next section we
collect together some definitions and preliminary results. In
\autoref{sec:hardness} we prove that \probWDB and \probWDCIS are strongly
\classNP-hard. In \autoref{sec:divBases} we derive our \classFPT and
kernelization algorithms for \probWDB, and in \autoref{sec:divComIndSets} we
show that \probWDCIS is \classFPT. We derive our results for \probPMD in
\autoref{sec:matchings}. We conclude in \autoref{sec:conclusion}.

\section{Preliminaries}\label{sec:prelim}  
We use \(X \sd Y\) to denote the \emph{symmetric difference} \((X \setminus
Y)\cup(Y \setminus X)\) of sets \(X\) and \(Y\). We use \naturals to denote the
set of positive integers.

\paragraph*{Parameterized complexity.}  A parameterized 
problem $\Pi$ is a subset of $\Sigma^*\times {\mathbb N}$, 
where $\Sigma$ is a finite alphabet. We say that a parameterized problem $\Pi$
is {\em fixed parameter tractable (\classFPT)}, if there is an algorithm that given an instance $(x,k)$  of $\Pi$ as input, solves in time $f(k)\vert x\vert^{\Oh(1)}$,  where $f$ is an arbitrary function and $\vert x\vert$ is the length of $x$.  A kernelization algorithm for a parameterized problem $\Pi$ is a polynomial time algorithm
(computable function) ${\cal A}~:~ \Sigma^*\times {\mathbb N}\rightarrow 
\Sigma^*\times {\mathbb N}$ such that $(x,k)\in \Pi$ if and only if 
$(x',k')={\cal A}((x,k))\in \Pi$ and $\vert x'\vert +k'\leq g(k)$ for some computable function $g$. 
When $g$ is a polynomial function, we say that $\Pi$ admits a polynomial kernel. 
For a detailed overview about parameterized complexity 
we refer to the monographs~\cite{DowneyF13,PCBookCygenetal,fomin_lokshtanov_saurabh_zehavi_2019}

\paragraph*{Matroids.} We give a brief description of the matroid-related
notions that we need. See the book of Oxley~\cite{Oxley92} for a detailed
introduction to matroids. A pair $M=(E,\mathcal{I})$, where $E$ is a finite
\emph{ground} set and $\mathcal{I}$ is a family of subsets of the ground set,
called \emph{independent sets} of $E$, is a \emph{matroid} if it satisfies the
following conditions, called \emph{independence axioms}:
\begin{description}
\item[(I1)]  $\emptyset \in \mathcal{I}$. 
\item[(I2)]  If $A\subseteq B\subseteq E(M) $ and $B\in \mathcal{I}$ then $A\in\mathcal{I}$. 
\item[(I3)] If $A, B  \in \mathcal{I}$  and $ |A| < |B| $, then there is $ e \in  B \setminus A$  such that $A\cup\{e\} \in \mathcal{I}$.
\end{description}
We use $E(M)$ and $\mathcal{I}(M)$ to denote the ground set and the set of
independent sets, respectively. As is standard for matroid problems, we assume
that each matroid \(M\) that appears in the
input 
is given by an \emph{independence oracle}, that is, an oracle that in constant
(or polynomial) time replies whether a given $A\subseteq E(M)$ is independent in
\(M\) or not. An inclusion-wise maximal independent set $B$ is called a
\emph{basis} of $M$. We use $\mathcal{B}(M)$ to denote the set of bases of $M$.
The bases satisfy the following properties, called \emph{basis axioms}:
\begin{description}
\item[(B1)]  $\mathcal{B}(M)\neq\emptyset$. 
\item[(B2)]  If $B_1,B_2\in \mathcal{B}(M)$, then for every $x\in B_1\setminus B_2$, there is $y\in B_2\setminus B_1$ such that $(B_1\setminus\{x\})\cup\{y\}\in \mathcal{B}(M)$. 
\end{description}
All the bases of $M$ have the same size that is called the \emph{rank} of $M$,
denoted \(\rank(M)\). The \emph{rank} of a subset $A\subseteq E(M)$, denoted
$\rank(A)$, is the maximum size of an independent set $X\subseteq A$; the
function $\rank\colon 2^{E(M)}\rightarrow \mathbb{Z}$ is the \emph{rank}
function of $M$. A set $A\subseteq E(M)$ \emph{spans} an element $x\in E(M)$ if
$\rank(A\cup\{x\})=\rank(A)$. The \emph{closure} (or \emph{span}) of $A$ is the
set $\cl(A)=\{x\in E(M)\mid A\text{ spans }x\}$. Closures satisfy the following
properties, called \emph{closure axioms}:
\begin{description}
\item[(CL1)]  For every $A\subseteq E(M)$, $A\subseteq \cl(A)$. 
\item[(CL2)]  If $A\subseteq B\subseteq E(M)$, then $\cl(A)\subseteq \cl(B)$.
\item[(CL3)]  For every $A\subseteq E(M)$, $\cl(A)=\cl(\cl(A))$.
\item[(CL4)]  For every $A\subseteq E(M)$ and every $x\in E(M)$ and $y\in\cl(A\cup\{x\})\setminus \cl(A)$, $x\in \cl(A\cup\{y\})$. 
\end{description}

The {\em dual} of a matroid $M=(E,\mathcal{I})$, denoted \MStar, is the matroid
whose ground set is \(E\) and whose set of bases is
$\mathcal{B}^*=\{\overline{B}\mid B\in \mathcal{B}(M)\}$. That is, the bases of
$M^*$ are exactly the complements of the bases of $M$. A basis (independent set,
rank, respectively) of $M^*$ is a \emph{cobasis} (\emph{coindependent set},
\emph{corank}, respectively) of $M$. We use $\mathcal{I}^*(M)$ to denote the set
of coindependent sets of $M$. Also, $\corank(M)$ denotes the corank of $M$ and
$\corank(A)$ denotes the corank of a set $A\subseteq E(M)$; $\corank(A)$ is the
rank of set \(A\) in the dual matrix \(M^{*}\), and \(\corank(A) = |A| -
\rank(M) + \rank(E \setminus A)\). Given an independence oracle for \(M\) we can
construct---using the augmentation property \dl{(I3)} and with an overhead which
is polynomial in \(|E|\)---a \emph{rank} oracle for \(M\), and thence
\emph{corank} and \emph{coindependence} oracles for \(M\).

For $e\in E(M)$, the matroid $M'=M-e$ is obtained by \emph{deleting} $e$ if
$E(M')=E(M)\setminus\{e\}$ and $\mathcal{I}(M')=\{X\in \mathcal{I}(M)\mid
e\notin X\}$. It is said that $M'=M/e$ is obtained by \emph{contracting} $e$ if
$M'=(\MStar-e)^*$. In particular, if $e$ is not a \emph{loop} (i.e., if
\(\{e\}\) is independent) in \(M\), then $\mathcal{I}(M')=\{X\setminus\{e\}\mid
e\in X\in \mathcal{I}(M)\}$. Notice that deleting an element in $M$ is
equivalent to contracting it in $M^*$ and vice versa. Let $X\subseteq E(M)$.
Then $M-X$ denotes the matroid obtained from $M$ by the deletion of the elements
of $X$ and $M/X$ is the matroid obtained by consecutive contractions of the
elements of $X$. Note that an independence oracle for \(M\) can itself act as an
independence oracle for \(M-X\) if we restrict our queries to subsets of \(E(M)
\setminus X\). Let \(\rank_{M/X}\) denote the rank function of the matroid
\(M/X\). Then for any \(Y\subseteq (E(M) \setminus X)\) we have that
\(\rank_{M/X}(Y) = \rank(X \cup Y) - \rank(X)\)~\cite[3.1.7]{Oxley92}. Given an
independence oracle for \(M\) we can thus easily construct an independence
oracle for \(M/X\).

Let $M$ be a matroid and let $\mathbb{F}$ be a field. An $n\times m$-matrix
$\bfA$ over $\mathbb{F}$ is a \emph{representation of $M$ over $\mathbb{F}$} if
there is one-to-one correspondence $f$ between $E(M)$ and the set of columns of
$\bfA$ such that for any $X\subseteq E(M)$, $X\in \mathcal{I}(M)$ if and only if
the columns $f(X)$ are linearly independent (as vectors of $\mathbb{F}^n$); if
$M$ has such a representation, then it is said that $M$ has a
\emph{representation over $\mathbb{F}$}. In other words, $\bfA$ is a
representation of $M$ if $M$ is isomorphic to the \emph{linear matroid} of
$\bfA$, i.e., the matroid whose ground set is the set of columns of $\bfA$ and a
set of columns is independent if and only if these columns are linearly
independent. Observe that, given a representation $\bfA$ of $M$, we can verify
whether a set is independent by checking the linear independence of the
corresponding columns of $\bfA$. Hence, we don't need an explicit independence
oracle in this case.

Let $1\leq r\leq n$ be integers. We use $U_n^r$ to denote the \emph{uniform} matroid, that is, the matroid with the ground set of size $n$ such that the bases are all $r$-element subsets of the ground set.

\medskip
We use the classical results of Edmonds~\cite{Edmonds70} and Frank~\cite{Frank81} about the \textsc{Weighed Matroid Intersection} problem. The task of this problem is, given two matroids $M_1$ and $M_2$ with the same ground set $E$ and a weight function $\omega\colon E\rightarrow \mathbb{N}$, find a set $X$ of maximum weight such that $X$ is independent in both matroids. Edmonds~\cite{Edmonds70} proved that the problem can be solved in polynomial time
for the unweighted case (that is, the task is to find a common independent set of maximum size; we refer to this variant as \textsc{Matroid Intersection}) and the result was generalized for the variant with the weights by Frank in~\cite{Frank81}. 

\begin{proposition}[\cite{Edmonds70,Frank81}]\label{prop:intersection}
\textsc{Weighted Matroid Intersection} can be solved in polynomial time.
\end{proposition}

We also need another classical result of Edmonds~\cite{Edmonds71} that a basis of maximum weight can be found by the greedy algorithm. Recall that, given a matroid $M$ with a weight function $\omega\colon E(M)\rightarrow \mathbb{N}$, the greedy algorithm finds a basis $B$ of maximum weight as follows. Initially, $B:=\emptyset$. Then at each iteration, the algorithm finds an element of $x\in E(M)\setminus B$ of maximum weight such that $B\cup \{x\}$ is independent  and sets $B:=B\cup\{x\}$. The algorithms stops when there is no element that can be added to $B$.

\begin{proposition}[\cite{Edmonds71}]\label{prop:greedy}
The greedy algorithm finds a basis of maximum weight of a weighted matroid in polynomial time. 
\end{proposition}

We need the following observation(See~\cite[Lemma~2.1.10]{Oxley92}).

\begin{observation}\label{obs:ind-coind}
  Let $X$ and $Y$ be disjoint sets such that $X$ is independent and $Y$ is
  coindependent in a matroid $M$. Then there is a basis $B$ of $M$ such that
  $X\subseteq B$ and $Y\cap B=\emptyset$.
\end{observation}


Observe that for any sets $X$ and $Y$ that are subsets of the same universe, $X\sd Y=\overline{X}\sd\overline{Y}$. This implies the following.

\begin{observation}\label{obs:dual}
  For every matroid $M$, every weight function $\omega\colon
  E(M)\rightarrow\mathbb{N}$, and all integers $k\geq 1$ and $d\geq 0$, the
  instances $(M,\omega,k,d)$ and $(M^*,\omega,k,d)$ of \probWDB are equivalent.
\end{observation}


\section{Hardness of \probWDB and \probWDCIS}\label{sec:hardness}
We show that \probWDB and \probWDCIS are \classNP-complete in the strong sense even for uniform matroids. 

\nphardness*

\begin{proof}

  We prove the claim for \probWDB by a reduction from the \textsc{$3$-Partition}
  problem. The input to \textsc{$3$-Partition} consists of a positive integer
  \(b\) and a multiset $S=\{s_1,\ldots,s_{3n}\}$ of $3n$ positive integers such
  that (i) $\frac{b}{4} < s_{i}< \frac{b}{2}$ holds for each
  $i\in\{1,\ldots,3n\}$ and (ii) \(\sum_{i=1}^{3n}s_{i} = nb\). The task is to
  decide whether $S$ can be partitioned into $n$ multisets $S_1,\ldots,S_n$ such
  that \(\sum_{s \in S_{i}}s = b\) holds for each \(S_{i}\). Note that each
  multiset \(S_{i}\) in such a partition must contain exactly three elements
  from \(S\). This problem is known to be \classNP-complete in the strong sense,
  i.e., it is \classNP-complete even if the input integers are encoded in
  unary~\cite[SP15]{GareyJ79}.

  Let $(b, S=\{s_1,\ldots,s_{3n})\}$ be an instance of \textsc{$3$-Partition}
  with $n\geq 3$. We set \(M\) to be the uniform matroid $U_{3n}^3$ on the
  ground set $\{1,\ldots,3n\}$, and define the weight function to be
  $\omega(i)=s_i$ for $i\in\{1,\ldots,3n\}$. We set $d=2b$. We will now show
  that $(b, S)$ is a yes-instance of \textsc{$3$-Partition} if and only if
  $(M,\omega,n,d)$ is a yes-instance of \probWDB.

  In the forward direction, suppose that $S_1,\ldots,S_n$ is a partition of $S$
  into triples of integers such that the sum of elements of each $S_i$ is $b$.
  Let $B_1,\ldots,B_n$ be the corresponding partition of $\{1,\ldots,3n\}$, that
  is, $B_i=\{i_1,i_2,i_3\}$ if and only if $S_i=\{s_{i_1},s_{i_2},s_{i_3}\}$ for
  each $i\in\{1,\ldots,n\}$. Clearly, $B_1,\ldots,B_n$ are pairwise disjoint
  bases of $M$. Then for every distinct $i,j\in\{1,\ldots,n\}$, $\omega(B_i\sd
  B_j) = \omega(B_{i}) + \omega(B_{j}) = 2b$. Therefore, $(M,\omega,n,d)$ is a
  yes-instance of \probWDB.

  In the reverse direction, assume that $(M,\omega,n,d)$ is a yes-instance of
  \probWDB. Let $B_1,\ldots,B_n$ be bases of $M$ such that $\omega(B_i\sd
  B_j)\geq d=2b$ for distinct $i,j\in\{1,\ldots,n\}$.

  We claim that $B_1,\ldots,B_n$ are pairwise disjoint. For the sake of
  contradiction, assume that there are distinct $i,j\in\{1,\ldots,n\}$ such that
  $B_i\cap B_j\neq\emptyset$. Let $X_1=B_i\setminus B_j$ and $X_2=B_j\setminus
  B_i$. Note that $|X_1|\leq 2$ and $|X_2|\leq 2$. We have that
  $\omega(X_1)=\sum_{h\in X_1}\omega(h)=\sum_{h\in X_1}s_h<|X_1|b/2 < b$.
  Similarly, $\omega(X_2)<b$. Therefore, $\omega(B_i\sd
  B_j)=\omega(X_1)+\omega(X_2)<2b$; a contradiction. We conclude that the bases
  $B_1,\ldots,B_n$ are pairwise disjoint. This implies that $B_1,\ldots,B_n$ is
  a \emph{partition} of \(\{1, 2, \dotsc, 3n\}\).

  Next we show that $\omega(B_i)=b$ holds for every $i\in\{1,\ldots,n\}$.
  Suppose that there is an $h\in\{1,\ldots,n\}$ such that $\omega(B_h)>b$. Let
  $I=\{1,\ldots,3n\}\setminus B_h$ and $J=\{1,\ldots,n\}\setminus\{h\}$. We have
  that $\sum_{i\in I}\omega(i)<\sum_{i=1}^{3n}\omega(i)-b=b(n-1)$.
  Since $B_1,\ldots,B_{h-1},B_{h+1}\ldots,B_n$ form a partition of $I$, we get
  that $\sum_{i\in J}\omega(B_i)<b(n-1)$ holds as well. Recall that $n\geq 3$.
  Then
  \begin{equation*}
    \sum_{ \{i,j\}\text{ s.t. }i,j\in J,~i\neq j}(\omega(B_i)+\omega(B_j)) \leq (n-2)\sum_{i\in J}\omega(B_i)<b(n-1)(n-2).
  \end{equation*}
  The first inequality above comes from the fact that since \(|J| = n-1\), for
  each index \(i \in J\) the term \(\omega(B_{i})\) appears in \emph{at most}
  \(n-2\) terms of the form \((\omega(B_i)+\omega(B_j))\) in the summation on
  the left hand side. Now suppose \((\omega(B_i)+\omega(B_j)) \geq 2b\) holds
  for \emph{all} pairs \(i,j \in J,\;i\neq j\). Then the sum on the left hand
  side would be at least \(\binom{|J|}{2} \cdot 2b = (n-1)(n-2)b\), a
  contradiction. Therefore, there must exist distinct $i,j\in J$ such that
  $\omega(B_i)+\omega(B_j)<2b$ holds. And this contradicts our assumption that
  $\omega(B_i\sd B_j)\geq 2b$ holds for all such
  \(i,j\). 
  We conclude that $\omega(B_i)\leq b$ holds for every $i\in \{1,\ldots,n\}$.
  And since $\sum_{i=1}^n\omega(B_i)=bn$, we get that $\omega(B_i)=b$ holds for
  every $i\in\{1,\ldots,n\}$.

  Finally, we consider the partition $S_1,\ldots,S_n$ of $S$ corresponding to
  $B_1,\ldots,B_i$, that is, for each $B_i=\{i_1,i_2,i_3\}$, we define
  $S_i=\{s_{i_1},s_{i_2},s_{i_3}\}$. Clearly,
  $s_{i_1}+s_{i_2}+s_{i_3}=\omega(B_i)=b$. Thus we get that $(b, S)$ is a
  yes-instance of \textsc{$3$-Partition}. This concludes the proof for \probWDB.

  The reduction for \probWDCIS is also from \textsc{$3$-Partition}, and is
  nearly identical to the above reduction for \probWDB. Given an instance $(b,
  S=\{s_1,\ldots,s_{3n})\}$ of \textsc{$3$-Partition} with $n\geq 3$, we set
  each of \(M_{1}, M_{2}\) to be the uniform matroid $U_{3n}^3$ on the ground
  set $\{1,\ldots,3n\}$, and define the weight function to be $\omega(i)=s_i$
  for $i\in\{1,\ldots,3n\}$. We set $d=2b$. We will now show that $(b, S)$ is a
  yes-instance of \textsc{$3$-Partition} if and only if
  $(M_{1},M_{2},\omega,n,d)$ is a yes-instance of \probWDCIS.

  In the forward direction, suppose that $S_1,\ldots,S_n$ is a partition of $S$
  into triples of integers such that the sum of elements of each $S_i$ is $b$.
  Let $I_1,\ldots,I_n$ be the corresponding partition of $\{1,\ldots,3n\}$, that
  is, $I_i=\{i_1,i_2,i_3\}$ if and only if $S_i=\{s_{i_1},s_{i_2},s_{i_3}\}$ for
  each $i\in\{1,\ldots,n\}$. Clearly, $I_1,\ldots,I_n$ are pairwise disjoint
  common independent sets of $M_{1}$ and \(M_{2}\), and for every distinct
  $i,j\in\{1,\ldots,n\}$, $\omega(I_i\sd I_j) = \omega(I_{i}) + \omega(I_{j}) =
  2b$. Therefore, $(M_{1},M_{2},\omega,n,d)$ is a yes-instance of \probWDCIS.

  In the reverse direction, assume that $(M_{1},M_{2},\omega,n,d)$ is a
  yes-instance of \probWDCIS, and let $I_1,\ldots,I_n$ be common independent
  sets of $M_{1}$ and \(M_{2}\) such that $\omega(I_i\sd I_j)\geq d=2b$ for
  distinct $i,j\in\{1,\ldots,n\}$. Since every independent set in the matroids
  \(M_{1},M_{2}\) has at most three elements, and since $s_{i}< \frac{b}{2}$
  holds for each $i\in\{1,\ldots,3n\}$, we get that the sets $I_1,\ldots,I_n$
  are pairwise disjoint. If two of these sets, say \(I_{i},I_{j}\) have at most
  two elements each then \(\omega(I_i\sd I_j) < 4\cdot \frac{b}{2} = 2b\), a
  contradiction. So at most one of these sets has at most two elements; every
  other set in the collection has exactly three elements.

  If all the sets $I_1,\ldots,I_n$ have three elements each then they are a
  pairwise disjoint collection of \(n\) \emph{bases} of $M_{1}$, and the
  argument that we used for the reverse direction in the proof for \probWDB
  tells us that $(b, S)$ is a yes-instance of \textsc{$3$-Partition}. In the
  remaining case there is exactly one set of size two among $I_1,\ldots,I_n$;
  without loss of generality, let this smaller set be \(I_{1}\). Then
  \(|\bigcup_{i=1}^{n}I_{i}| = 3n-1\). Let \(x = \{1, 2, \dotsc, 3n\} \setminus
  \bigcup_{i=1}^{n}I_{i}\) be the unique element which is not in any of these
  independent sets. Then \((I_{1} \cup \{x\}),\ldots,I_n\) is a pairwise
  disjoint collection of \(n\) bases of $M_{1}$ such that the weight of the
  symmetric difference of any pair of these bases is at least \(d=2b\), and the
  argument that we used for the reverse direction in the proof for \probWDB
  tells us that $(b, S)$ is a yes-instance of \textsc{$3$-Partition}.
\end{proof}


\section{An FPT algorithm and kernelization for \probWDB}\label{sec:divBases}
In this section, we show that \probWDB is \classFPT when parameterized by $k$ and $d$. Moreover, if the input matroid is representable over a finite field and is given by such a representation, then \probWDB admits a polynomial kernel. 

We start with the observation that if the input matroid has a sufficiently big
set that is simultaneously independent and coindependent, then diverse bases
always exist.

\begin{lemma}\label{lem:ind-coind}
  Let $M$ be a matroid, and let $k\geq 1$ and $d\geq 0$ be integers. If there is
  $X\subseteq E(M)$ of size at least $k\lceil\frac{d}{2}\rceil$ such that $X$ is
  simultaneously independent and coindependent, then $(M,\omega,k,d)$ is a
  yes-instance of \probWDB for any weight function $\omega$.
\end{lemma}

\begin{proof}
  Let $X\subseteq E(M)$ be a set of size at least $k\lceil\frac{d}{2}\rceil$
  such that $X$ is simultaneously independent and coindependent. Then there is a
  partition $X_1,\ldots,X_k$ of $X$ such that $|X_i|\geq
  \lceil\frac{d}{2}\rceil$ for every $i\in\{1,\ldots,k\}$. Let
  $i\in\{1,\ldots,k\}$. Since $X$ is independent, $X_i$ is independent, and
  since $X$ is coindependent, then $X\setminus X_i$ is coindependent. Then by
  \autoref{obs:ind-coind}, there is a basis $B_i$ of $M$ such that $X_i\subseteq
  B_i$ and $B_i\cap(X\setminus X_i)=\emptyset$. The latter property means that
  $B_i\cap X_j=\emptyset$ for every $j\in \{1,\ldots,k\}$ such that $j\neq i$.
  We consider the bases $B_i$ defined in this manner for all
  $i\in\{1,\ldots,k\}$. Then for every distinct $i,j\in\{1,\ldots,k\}$, $X_i\cup
  X_j\subseteq B_i\sd B_j$. Therefore, $\omega(B_i\sd B_j)\geq \omega(X_i\cup
  X_j)\geq |X_i\cup X_j|=|X_i|+|X_j|\geq 2\lceil\frac{d}{2}\rceil\geq d$ for any
  $\omega\colon E(M)\rightarrow\mathbb{N}$. Hence, $(M,\omega,k,d)$ is a
  yes-instance of \probWDB.
 \end{proof}

Our results are based on the following lemma.

\begin{lemma}\label{lem:main}
  There is an algorithm that, given an instance $(M,\omega,k,d)$ of \probWDB,
  runs in time polynomial in \((|E(M)| + k + d)\) and either correctly decides
  that $(M,\omega,k,d)$ is a yes-instance or outputs an equivalent instance
  $(\widetilde{M},\omega,k,d)$ of \probWDB such that $E(\widetilde{M})\subseteq
  E(M)$ and $|E(\widetilde{M})|\leq 2\lceil\frac{d}{2}\rceil^2k^3$. In the
  latter case, the algorithm also computes a partition $(L,L^*)$ of
  $E(\widetilde{M})$ with the property that for every basis $B$ of
  $\widetilde{M}$, $|B\cap L|\leq\lceil\frac{d}{2}\rceil k$ and $|L^*\setminus
  B|\leq \lceil\frac{d}{2}\rceil k$, and the algorithm outputs an independence
  oracle for $\widetilde{M}$ that answers queries for $\widetilde{M}$ in time
  polynomial in $|E(M)|$. Moreover, if $M$ is representable over a finite field
  $\mathbb{F}$ and is given by such a representation, then the algorithm outputs
  a representation of $\widetilde{M}$ over $\mathbb{F}$.
\end{lemma}
  
\begin{proof}
  Let $(M,\omega,k,d)$ be an instance of \probWDB. Recall that \(M\) is given as
  an independence oracle. We construct an independence oracle for the dual
  matroid $M^*$, and then solve \textsc{Matroid Intersection} for $M$ and $M^*$
  using \autoref{prop:intersection}. Let $X$ be the set computed by the
  \textsc{Matroid Intersection} algorithm. Then $X\subseteq E(M)$ is a set of
  maximum size that is both independent and coindependent in $M$. If $|X|\geq
  k\lceil\frac{d}{2}\rceil$, then $(M,\omega,k,d)$ is a yes-instance of \probWDB
  by \autoref{lem:ind-coind}; the problem is solved and we return the answer.

  Assume from now on that this is not the case, and that $|X|\leq
  k\lceil\frac{d}{2}\rceil-1$ holds. Let \(B\) be an arbitrary basis of \(M\),
  and let \(\overline{B} = (E(M) \setminus B)\). If \(\rank(\overline{B})\geq
  k\lceil\frac{d}{2}\rceil\) then there exists an independent set \(Y\subseteq
  \overline{B}\) of size at least \(k\lceil\frac{d}{2}\rceil\). But \(Y\) is
  also a \emph{coindependent} set of size at least \(k\lceil\frac{d}{2}\rceil\),
  which contradicts our assumption. Thus we get that $\rank(\overline{B})\leq
  k\lceil\frac{d}{2}\rceil-1$ and $\corank(B)\leq k\lceil\frac{d}{2}\rceil-1$
  hold for any basis $B$ of $M$.

  Let $\ell=\lceil\frac{d}{2}\rceil k^2$. Fix an arbitrary basis $B$ of $M$. We
  construct sets $S_0,\ldots,S_\ell$ iteratively. We set $S_0=B$. For $i\geq 1$
  we construct $S_i$ from \(S_{(i-1)}\) as follows. If $E(M)\setminus
  S_{i-1}=\emptyset$, we set $X_i=\emptyset$. Otherwise we set \(X_{i}\) to be a
  basis of maximum weight in the matroid $M-S_{i-1}$; we find \(X_{i}\) using
  the greedy algorithm (see \autoref{prop:greedy}). Finally, we set
  $S_i=S_{i-1}\cup X_i$.

  Let $S=S_\ell$ and $L=S\setminus B$. Since $\rank(\overline{B})\leq
  k\lceil\frac{d}{2}\rceil-1$, we get that \emph{every} independent set
  contained in the set \(\overline{B} = (E(M) \setminus B)\) has size at most
  \(k\lceil\frac{d}{2}\rceil-1\). And since \(L\) is a disjoint union of
  \(\ell\) such independent sets we get that $|L|=|S\setminus B|\leq \ell
  (k\lceil\frac{d}{2}\rceil-1)\leq \lceil\frac{d}{2}\rceil^2k^3$. We show the
  following crucial claim.

\begin{claim}\label{cl:kern}
If $(M,\omega,k,d)$ is a yes-instance of \probWDB, then there is a solution, that is, a family of bases $B_1,\ldots,B_k$ such that  $\omega(B_i \sd B_j)\geq d$ for all distinct $i,j\in\{1,\ldots,k\}$, with the property that $B_i\subseteq S$ for every $i\in\{1,\ldots,k\}$. 
\end{claim}

\begin{claimproof}
  Let $(M,\omega,k,d)$ be a yes-instance, and let the family of bases
  $B_1,\ldots,B_k$ be a solution which maximizes the size of the set
  \(((\bigcup_{i=1}^{k}B_{i}) \cap S)\) of vertices in the bases which are also
  in the set $S$. We show that $B_i\subseteq S$ holds for every
  $i\in\{1,\ldots,k\}$. The proof is by contradiction. Assume that there is an
  $h\in\{1,\ldots,k\}$ such that $B_h\setminus S\neq\emptyset$. Recall that
  $\rank(M-B)=\rank(\overline{B})\leq k\lceil\frac{d}{2}\rceil-1$. Therefore,
  $|B_i\setminus B|\leq k\lceil\frac{d}{2}\rceil-1$ holds for every
  $i\in\{1,\ldots,k\}$, and $|\cup_{i=1}^k(B_i\setminus B)|\leq
  k(k\lceil\frac{d}{2}\rceil-1)<\ell$. Let $X_1,\ldots,X_\ell$ be the
  independent sets used to construct the sets $S_1,\ldots,S_\ell$. Since
  \((X_{j} \cap B) =\emptyset\) holds for all $j\in\{1,\ldots, \ell\}$ we get
  that \((X_{j} \cap B_{i}) = (X_{j} \cap (B_{i} \setminus B))\) holds for all
  \(i\in\{1,\ldots, k\}, j\in\{1,\ldots, \ell\}\). So the number
  \(|\cup_{i=1}^k(B_i\setminus B)|\) of elements from the bases $B_1,\ldots,B_k$
  which could \emph{potentially} be part of any of the sets $X_1,\ldots,X_\ell$
  is strictly less than the number of these latter sets. Hence from the
  pigeonhole principle we get that there is a $t\in\{1,\ldots,\ell\}$ such that
  $X_t\cap B_i=\emptyset$ holds for all $i\in\{1,\ldots,k\}$. Let $A=B_h\cap
  S_{(t-1)}$ and 
  $Y=B_h\setminus S_{(t-1)}$. We show that there is $Z\subseteq X_t$ such that
  \begin{itemize}
  \item[(i)] $B_h'=A\cup Z$ is a basis, and
  \item[(ii)] $\omega(Z)\geq \omega(Y)$.
  \end{itemize}

  We construct $Z$ by greedily augmenting $A$ with elements of $X_t$. Let
  \(\sigma\) be the order in which the greedy algorithm picks elements from the
  set \(E(M) \setminus S_{(t-1)}\) to add them to the set \(X_{t}\). Initially
  we set $Z:=\emptyset$. Then we select the first $x\in X_t\setminus Z$ in
  $\sigma$ such that $A\cup Z\cup \{x\}$ is independent, and we set
  $Z:=Z\cup\{x\}$. We stop when there is no $x\in X_t\setminus Z$ such that
  $A\cup Z\cup \{x\}$ is independent. We prove that (i) and (ii) are fulfilled
  for $Z$.

  First we show that (i) holds. From the construction we get that \(X_{t}\) is a
  basis of the matroid \(M-S_{(t-1)}\). This implies that \((E(M) \setminus
  S_{(t-1)}) \subseteq \cl(X_{t})\) holds. Now since \(Y = (B_{h} \setminus
  S_{(t-1)})\) is a subset of \((E(M) \setminus S_{(t-1)})\) we get that
  $Y\subseteq \cl(X_t)$ holds. Since \((A \cup Z)\) is independent, and there is
  no $x\in X_t\setminus Z$ such that $A\cup Z\cup \{x\}$ is independent, we get
  that $X_t\subseteq \cl(A\cup Z)$ holds. Now by \dl{(CL2)} and \dl{(CL3)},
  $Y\subseteq \cl(\cl(A\cup Z))=\cl(A\cup Z)$. And by \dl{(CL1)}, $A\subseteq
  \cl(A\cup Z)$ and we conclude that $A\cup Y\subseteq\cl(A\cup Z)$ holds. But
  \((A \cup Y) = B_{h}\) is a \emph{basis} of \(M\), and
  so~\cite[Proposition~1.4.9]{Oxley92} \(\cl(A \cup Y) = E(M)\). Applying
  \dl{(CL2)} and \dl{(CL3)} we get that \(E(M) \subseteq \cl(A\cup Z)\) which
  implies that \(\cl(A\cup Z) = E(M)\). Now since \(A\cup Z\) is independent we
  get (See, e.g., \cite[Section~1.4, Exercise~2]{Oxley92}) that $B_h'=A\cup Z$
  is a basis.

  Now we show that (ii) holds. Let $Z=\{z_1,\ldots,z_s\}$, where the elements
  are indexed according to the order in which they are added to \(Z\) by the
  greedy augmentation described above. Note that $\omega(z_1)\geq\cdots\geq
  \omega(z_s)$. Since $B_h$ and $B_h'$ are bases, \(Y=(B_{h} \setminus
  S_{(t-1)})\), and \(Z=(B'_{h} \setminus S_{(t-1)})\), we get that $|Y|=|Z|$.
  Observe also that \(Y \cap Z = \emptyset\). We define (i) $Z_0=\emptyset$ and
  (ii) $Z_i=\{z_1,\ldots,z_i\}$ for $i\in\{1,\ldots,s\}$. We show that there is
  an ordering \(\langle y_1,\ldots,y_s \rangle\) of the elements of \(Y\) such
  that the set $A\cup Z_{i-1}\cup\{y_i\}$ is independent for every
  $i\in\{1,\ldots,s\}$. We define this order inductively, starting with
  \(y_{s}\) and proceeding in decreasing order of the subscript.

  We set \(y_{s}\) to be an element \(y\in (A\cup Y)\setminus (A\cup Z_{s-1}) =
  (Y \setminus Z_{s-1})\) such that $A\cup Z_{s-1}\cup \{y\}$ is independent.
  Since $|A\cup Z_{s-1}|<|A\cup Y|$ we know from \dl{(I3)} such an element must
  exist. For the inductive step, assume that for some fixed $i\in
  \{1,\ldots,s-1\}$ distinct elements $y_{i+1},\ldots,y_s\in Y$ have been
  defined such that $A\cup Z_i\cup\{y_{i+1},\ldots,y_s\}$ is independent. Note
  that \(|A\cup Z_i\cup\{y_{i+1},\ldots,y_s\}| = |A\cup Y|\). Then $R=A\cup
  Z_{i-1}\cup\{y_{i+1},\ldots,y_s\}$ is independent by \dl{(I2)}, and by
  \dl{(I3)}, there must exist an element $y\in (A\cup Y)\setminus R$ such that
  $R\cup\{y\}= A\cup Z_{i-1}\cup\{y, y_{(i+1)},\ldots,y_s\}$ is independent. We
  set \(y_{i}\) to be this element \(y\). Observe that due to \dl{(I2)}, $A\cup
  Z_{i-1}\cup\{y_i\}$ is indeed independent for every $i\in\{1,\ldots,s\}$.
  
  We claim that $\omega(y_i)\leq\omega(z_i)$ holds for every
  $i\in\{1,\ldots,s\}$. For the sake of contradiction, assume that this is not
  the case and let $i\in\{1,\ldots,s\}$ be the first index such that
  $\omega(y_i)>\omega(z_i)$ holds. Recall that $X_t$ is constructed by the
  greedy algorithm. Denote by $W\subset X_t$ the set of elements that are prior
  $z_i$ in the ordering $\sigma$. Suppose that $y_i\in\cl(W)$. By the
  construction of $Z$, $W\subseteq \cl(A\cup Z_{i-1})$, because $z_i$ is the
  first element in $\sigma$ such that $A\cup Z_{i-1}\cup\{z_i\}$ is independent.
  By \dl{(CL2)} and \dl{(CL3)}, we have that $y_i\in \cl(\cl(A\cup
  Z_{i-1}))=\cl(A\cup Z_{i-1})$. However, this contradicts the property that
  $A\cup Z_{i-1}\cup\{y_i\}$ is independent. Hence, $y_i\notin\cl(W)$. This
  implies that $W\cup \{y_i\}$ is independent. But this means that the greedy
  algorithm would have $y_i$ over $z_i$ in the construction of $X_t$, because
  $\omega(y_i)>\omega(z_i)$; a contradiction. This proves that
  $\omega(y_i)\leq\omega(z_i)$ holds for every $i\in\{1,\ldots,s\}$. Therefore,
  $\omega(Z)\geq \omega(Y)$ and (ii) is fulfilled. This completes the proof of
  the existence of a set $Z\subseteq X_t$ satisfying (i) and (ii).

  We replace the basis $B_h$ in the solution by $B_h'=A\cup Z=(B_h\setminus
  Y)\cup Z$. We show that the resulting family of bases is a solution to the
  instance $(M,\omega,k,d)$. Clearly, it is sufficient to show that for every
  $i\in\{1,\dots,k\}$ such that $i\neq h$, $\omega(B_h'\sd B_i)\geq d$, as the
  other pairs of bases are the same as before. By the choice of
  $t\in\{1,\ldots,\ell\}$ we have that $X_{t}\cap B_i=\emptyset$ holds for all
  $i\in\{1,\dots,k\}$. And since $Z\subseteq X_t$ we have that $Z\subseteq
  B_h'\sd B_i$ holds. Then, $\omega(B_h'\sd B_i)=\omega(((B_h\setminus Y)\cup
  Z)\sd B_i))\geq \omega(B_h\sd B_i)-\omega(Y)+\omega(Z)\geq\omega(B_h\sd
  B_i)\geq d$ as required. We have that the replacement of $B_h$ by $B_h'$ gives
  a solution. However $B_h'\subseteq S_t\subseteq S$ whereas \(B_{h}\setminus S
  \neq \emptyset\), and this contradicts the assumption that $B_1,\ldots,B_k$ is
  a solution such that the number of vertices of the bases in $S$ is the
  maximum. This concludes the proof of the claim.
\end{claimproof}

Let $\widehat{M}=M-(E(M)\setminus S)$. Then \(E(\widehat{M}) = S\) and the set
\(B\) is a basis of \(\widehat{M}\) as well. \autoref{cl:kern} immediately
implies the following property.

\begin{claim}\label{cl:first-step}
The instances $(M,\omega,k,d)$ and $(\widehat{M},\omega,k,d)$ of \probWDB are equivalent.
\end{claim}

We now repeat the argument that preceded \autoref{cl:kern}, this time with the
dual matroid \(\widehat{M}^{*}\) and starting with its basis \(\widehat{B} =
(E(\widehat{M})\setminus B) = L\). Recall that $\ell=\lceil\frac{d}{2}\rceil
k^2$. We construct sets \(S^{*}_{0},\ldots,S^{*}_{\ell}\) iteratively. We set
\(S^{*}_{0}=\widehat{B}\). For $i\geq 1$ we construct \(S^{*}_{i}\) from
\(S^{*}_{(i-1)}\) as follows. If $E(\widehat{M})\setminus
S^{*}_{i-1}=\emptyset$, we set $X^{*}_i=\emptyset$. Otherwise we set
\(X^{*}_{i}\) to be a basis of maximum weight in the matroid
$\widehat{M}^{*}-S^{*}_{i-1}$, which we find using the greedy algorithm.
Finally, we set $S^{*}_i=S^{*}_{i-1}\cup X^{*}_i$.

Let $S^{*}=S^{*}_\ell$ and \(L^{*} = S^{*}\setminus \widehat{B} = S^{*}\cap
B\). Since $\corank(B)\leq k\lceil\frac{d}{2}\rceil-1$, we get that \emph{every}
coindependent set contained in the set \(B\) has size at most
\(k\lceil\frac{d}{2}\rceil-1\). And since \(L^{*}\) is a disjoint union of
\(\ell\) such coindependent sets we get that \(|L^{*}|=|S^{*}\cap B|\leq \ell
(k\lceil\frac{d}{2}\rceil-1)\leq \lceil\frac{d}{2}\rceil^2k^3\). Restating
\autoref{cl:kern} for $\widehat{M}^*$, we get that if
$(\widehat{M}^*,\omega,k,d)$ is a yes-instance of \probWDB, then there is a
solution, that is, a family of bases $B_1^*,\ldots,B_k^*$ of $\widehat{M}^*$
such that $\omega(B_i^* \sd B_j^*)\geq d$ holds for all distinct
$i,j\in\{1,\ldots,k\}$, with the property that $B_i^*\subseteq S^*$ for every
$i\in\{1,\ldots,k\}$. In terms of $\widehat{M}$, the same property can be stated
as follows.

\begin{claim}\label{cl:cokern}
  If $(\widehat{M},\omega,k,d)$ is a yes-instance of \probWDB, then there is a
  solution, that is, a family of bases $B_1,\ldots,B_k$ such that $\omega(B_i
  \sd B_j)\geq d$ for all distinct $i,j\in\{1,\ldots,k\}$ with the property that
  $\overline{B_i}\subseteq S^*$ for every $i\in\{1,\ldots,k\}$, where
  \(\overline{B_{i}} = (E(\widehat{M}) \setminus B_{i})\).
\end{claim}

Since \(E(\widehat{M}) = (B \cup \widehat{B})\) and \(\widehat{B} \subseteq
S^{*} \subseteq E(\widehat{M})\) we have that \(E(\widehat{M}) = (B \cup
S^{*})\). Hence from \autoref{cl:cokern} we get that if
$(\widehat{M},\omega,k,d)$ is a yes-instance, then it has a solution
$B_1,\ldots,B_k$ such that $(B\setminus S^*)\subseteq B_i$ holds for every
$i\in\{1,\ldots,k\}$. That is, elements from the set \((B\setminus S^*)\) do not
contribute to the weight $\omega(B_i \sd B_j)$ for any distinct
$i,j\in\{1,\ldots,k\}$. So a transformation that removes the subset
\((B\setminus S^*)\) from the ground set of \(\widehat{M}\) is safe,
\emph{provided that} (i) \(B_{i} \setminus(B\setminus S^{*}) = (B_{i} \cap
S^{*})\) is a basis of the resulting matroid for all $i\in\{1,\ldots,k\}$, and
(ii) for \emph{any} basis \(B'\) of the resulting matroid, \(B' \cup (B\setminus
S^*)\) is a basis of \(\widehat{M}\).

  We now show that the operation of \emph{contracting} the set \((B\setminus
  S^{*})\) has both these properties. Let
  $\widetilde{M}=\widehat{M}/(B\setminus S^*)$. Then \(E(\widetilde{M}) = (B
  \cup S^{*}) \setminus (B\setminus S^*) = S^{*}\). Let \(\rank(\widehat{M}),
  \rank(\widetilde{M})\) be the ranks and \(\widehat{\rank}, \widetilde{\rank}\)
  be the rank functions of the two matroids \(\widehat{M}, \widetilde{M}\),
  respectively. Recall that \(\widetilde{\rank}(X) = \widehat{\rank}((B\setminus
  S^*) \cup X) - \widehat{\rank}(B\setminus S^*)\) holds for all \(X \subseteq
  E(\widetilde{M}) = S^{*}\). Now \(\rank(\widetilde{M}) =
  \widetilde{\rank}(S^{*}) = \widehat{\rank}(B \cup S^{*}) - \widehat{\rank}(B
  \setminus S^{*}) = \rank(\widehat{M}) - |B \setminus S^{*}|\), where the last
  equation holds because \(B\) is a basis of \(\widehat{M}\). And for any
  $i\in\{1,\ldots,k\}$, \(\widetilde{\rank}(B_{i} \cap S^{*}) =
  \widehat{\rank}((B \setminus S^{*}) \cup (B_{i} \cap S^{*})) -
  \widehat{\rank}(B \setminus S^{*}) = \rank(\widehat{M}) - |B \setminus S^{*}|
  = \rank(\widetilde{M})\), where the second equation holds because \(B, B_{i}\)
  are bases of \(\widehat{M}\) and \(B_{i} \subseteq ((B \setminus S^{*}) \cup
  (B_{i} \cap S^{*}))\). Thus \((B_{i} \cap S^{*})\) is a basis of
  \(\widetilde{M}\). Finally, let \(B'\) be an arbitrary basis of
  \(\widetilde{M}\). Then \(\widetilde{\rank}(B') = \rank(\widetilde{M}) =
  \rank(\widehat{M}) - |B \setminus S^{*}|\). Rearranging the expression for
  \(\widetilde{\rank}(B')\) in terms of \(\widehat{\rank}\) we get:
  \(\widehat{\rank}((B \setminus S^{*}) \cup B') = \widetilde{\rank}(B') +
  \widehat{\rank}(B \setminus S^{*}) = \rank(\widehat{M}) - |B \setminus S^{*}|
  + |B \setminus S^{*}| = \rank(\widehat{M})\) where the second equation holds
  because \(B\) is a basis of \(\widehat{M}\). Thus \((B \setminus S^{*}) \cup
  B'\) is a basis of \(\widehat{M}\), and we have

 \begin{claim}\label{cl:last-step}
The instances $(M,\omega,k,d)$ and $(\widetilde{M},\omega,k,d)$ of \probWDB are equivalent.
\end{claim}

Recall the sets \(L = \widehat{B} \subseteq S^{*}\) and \(L^{*} = (S^{*}
\setminus L)\) from the construction. $(L,L^*)$ is thus a partition of
$E(\widetilde{M}) = S^{*}$. From the construction we get \(L = \widehat{B}
\subseteq \overline{B}\) and \(L^{*} \subseteq B\). Now since
$\rank(\overline{B})\leq \lceil\frac{d}{2}\rceil k$ and
$\corank(B)\leq\lceil\frac{d}{2}\rceil k$ in $M$, we have that for every basis
$B'$ of $\widetilde{M}$, $|B'\cap L|\leq \lceil\frac{d}{2}\rceil k$ and $|L^*
\setminus B'|\leq \lceil\frac{d}{2}\rceil k$ hold.

This completes the description of the algorithm that returns the instance
$(\widetilde{M},\omega,k,d)$ and the partition $(L,L^*)$ of $E(\widetilde{M})$.
Since \(|L| \leq \lceil\frac{d}{2}\rceil^2k^3\) and \(|L^{*}| \leq
\lceil\frac{d}{2}\rceil^2k^3\), we have that $|E(\widetilde{M})|\leq
2\lceil\frac{d}{2}\rceil^2k^3$. It is straightforward to verify that given an
independence oracle for $M$ we can construct the following in polynomial time:
(i) the set $E(\widetilde{M})$, (ii) an independence oracle for $\widetilde{M}$
that in time polynomial in $|E(M)|$ answers queries for $\widetilde{M}$, and
(iii) the sets $L$ and
$L^*$. 
To see this, note that $A\subseteq E(\widetilde{M})$ is independent in
$\widetilde{M}$ if and only if $A'=A\cup(B\setminus S^*)$ is independent in $M$.
 
To show the second claim of the lemma, assume that we are given representation $\bfA$ of $M$ over a finite field $\mathbb{F}$. It is well-known that $M^*$ also is representable over $\mathbb{F}$ and, given $\bfA$, the representation of $M^*$ over $\mathcal{F}$ can be computed in polynomial time by linear algebra tools (see, e.g.,~\cite{Oxley92}). Taking into account that contraction of a set is equivalent to the deletion of the  same set in the dual matroid and vice versa, we obtain that the representation $\tilde{\bfA}$ of $\widetilde{M}$ can be constructed in polynomial time from $A$. This concludes the proof of the lemma.
\end{proof}

Using \autoref{lem:main} we can prove that \probWDB is \classFPT when
parameterized by $k$ and $d$.

\divBasesFPT*

\begin{proof}
  Let $(M,\omega,k,d)$ be an instance of \probWDB. We run the algorithm from
  \autoref{lem:main}. If the algorithm solves the problem, then we are done.
  Otherwise, the algorithm outputs an equivalent instance
  $(\widetilde{M},\omega,k,d)$ of \probWDB such that $E(\widetilde{M})\subseteq
  E(M)$ and $|E(\widetilde{M})|\leq 2\lceil\frac{d}{2}\rceil^2k^3$. Moreover,
  the algorithm computes the partition $(L,L^*)$ of $E(\widetilde{M})$ with the
  property that for every basis $B$ of $\widetilde{M}$, $|B\cap
  L|\leq\lceil\frac{d}{2}\rceil k$ and $|L^*\setminus B|\leq
  \lceil\frac{d}{2}\rceil k$. Then we check all possible $k$-tuples of bases by
  brute force and verify whether there are $k$ bases forming a solution. By the
  properties of $L$ and $L^*$, $\widetilde{M}$ has $(d^2k^3)^{\Oh(dk)}$ distinct
  bases. Therefore, we check at most $(d^2k^3)^{\Oh(dk^2)}$ $k$-tuples of bases.
  We conclude that this checking can be done in $2^{\Oh(dk^2(\log k+\log
    d))}\cdot |E(M)|^{\Oh(1)}$ time, and the claim follows.
\end{proof}

If the input matroid is given by a representation over a finite field, then  \probWDB admits a polynomial kernel when parameterized by $k$ and $d$.

\kernel*

\begin{proof}
  Let $(M,\omega,k,d)$ be an instance of \probWDB. Let also $\bfA$ be its
  representation over $\gfq$. We run the algorithm from \autoref{lem:main}. If
  the algorithm solves the problem and reports that $(M,\omega,k,d)$ is a
  yes-instance, we return a trivial yes-instance of the problem. Otherwise, the
  algorithm outputs an equivalent instance $(\widetilde{M},\omega,k,d)$ of
  \probWDB such that $E(\widetilde{M})\subseteq E(M)$ and
  $|E(\widetilde{M})|\leq 2\lceil\frac{d}{2}\rceil^2k^3$. Moreover, the
  algorithm computes a representation $\tilde{\bfA}$ of $\widetilde{M}$ over
  $\gfq$. Clearly, it can be assumed that the number of rows of the matrix
  $\tilde{\bfA}$ equals $\rank(\widetilde{M})$. Since $\rank(\widetilde{M})\leq
  |E(\widetilde{M})|$, the matrix $\tilde{\bfA}$ has $\Oh(k^6d^4)$ elements.
  Because $\tilde{\bfA}$ is a matrix over $\gfq$, it can be encoded by
  $\Oh(k^6d^4\log q)$ bits. Finally, note that the weights of the elements can
  be truncated by $d$, that is, we can set $\omega(e):=\min\{\omega(e),d\}$ for
  every $e\in E(\widetilde{M})$. Then the weights can be encoded using
  $\Oh(d^2k^3\log d)$ bits. This concludes the construction of our kernel.
\end{proof}


\section{An \classFPT algorithm for Weighted Diverse Common Independent
  Sets}\label{sec:divComIndSets}
In this section we show that \probWDCIS is \classFPT when parameterized by $k$
and $d$.

We use a similar win-win approach as for \probWDB and observe that if the two
matroids from an instance of \probWDCIS have a sufficiently big common
independent set, then we have a yes-instance of \probWDCIS.

\begin{lemma}\label{lem:big-is}
  Let $M_1$ and $M_2$ be matroids with a common ground set $E$, and let $k\geq
  1$ and $d\geq 0$ be integers. If there is an $X\subseteq E$ of size at least
  $k\lceil\frac{d}{2}\rceil$ such that $X$ is a common independent set of $M_1$
  and $M_2$, then $(M_1,M_2,\omega,k,d)$ is a yes-instance of \probWDCIS for any
  weight function $\omega:E\to \naturals$.
\end{lemma}

\begin{proof}
  Let $X\subseteq E$ be a set of size at least $k\lceil\frac{d}{2}\rceil$ such
  that $X$ is a common independent set of $M_1$ and $M_2$. Then there is a
  partition $I_1,\ldots,I_k$ of $X$ such that $|I_i|\geq
  \lceil\frac{d}{2}\rceil$ for every $i\in\{1,\ldots,k\}$. Clearly,
  $I_1,\ldots,I_k$ are common independent sets of $M_1$ and $M_2$. Also we have
  that $\omega(I_i\sd I_j)=\omega(I_i)+\omega(I_j)\geq d$ for all distinct
  $i,j\in\{1,\ldots,k \}$ and every weight function $\omega$ which assigns
  positive integral weights. This means that $I_1,\ldots,I_k$ is a solution for
  $(M_1,M_2,\omega,k,d)$; that is, $(M_1,M_2,\omega,k,d)$ is a yes-instance.
 \end{proof}

 \autoref{lem:big-is} implies that we can assume that the maximum size of a
 common independent set of the input matroids is bounded. We prove the following
 crucial lemma.

\begin{lemma}\label{lem:branching}
  Let $(M_1,M_2,\omega,k,d)$ be an instance of \probWDCIS such that the maximum
  size of a common independent set of $M_1$ and $M_2$ is at most $s$. Then there
  is a set $\mathcal{F}$ of common independent sets of $M_1$ and $M_2$, of size
  $|\mathcal{F}| = 2^{\Oh(s^2\log(ks))}\cdot d$, such that if
  $(M_1,M_2,\omega,k,d)$ is a yes-instance of \probWDCIS then the instance has a
  solution $I_1,\ldots,I_k$ with $I_i\in\mathcal{F}$ for $i\in\{1,\ldots,k\}$.
  Moreover, $\mathcal{F}$ can be constructed in $2^{\Oh(s^2\log(ks))}\cdot
  d\cdot |E|^{\Oh(1)}$ time where \(E\) is the (common) ground set of \(M_{1}\)
  and \(M_{2}\).
\end{lemma} 
 
\begin{proof}
  Consider $(M_1,M_2,\omega,k,d)$. Let $E=E(M_1)= E(M_2)$. It is convenient to
  assume that the weights of the elements are bounded by $d$. For this, we set
  $\omega(e):=\min\{d,\omega(e)\}$ for every $e\in E$. It is straightforward to
  see that by this operation we obtain an equivalent instance of \probWDCIS.
  Notice that for every common independent set $I$ of $M_1,M_2$, we now have
  $\omega(I)\leq ds$.

  For every $w\in\{0,\ldots,ds\}$, we use a recursive branching algorithm to
  construct a family $\mathcal{F}_w$ of size $2^{\Oh(s^2\log(ks))}$ of common
  independent sets of $M_1$ and $M_2$ with the following properties: (i) each
  set in $\mathcal{F}_w$ has weight \emph{at least} \(w\), and (ii) if
  $S=\{I_1,\ldots,I_k\}$ is a solution to the instance $(M_1,M_2,\omega,k,d)$
  such that $\omega(I_i)=w$ for some $i\in\{1,\ldots,k\}$, then there is an
  $I_i'\in \mathcal{F}_w$ such that \((S \setminus I_{i}) \cup I_i'\) is also a
  solution to $(M_1,M_2,\omega,k,d)$.

  The algorithm, denoted by $\mathcal{A}$, takes as its input a common
  independent set $X$ of $M_1$ and $M_2$, and two matroids $M_1'$ and $M_2'$
  such that $M_i'=(M_i-W)/X$ for $i=1,2$ for some subset $W\subseteq E\setminus
  X$. For the very first call to \(\mathcal{A}\) we set $X:=\emptyset$ and
  $M_i'=M_i$ for $i=1,2$ (thus we implicitly set \(W:=\emptyset\)). Algorithm
  $\mathcal{A}$ outputs at most \(ks\) common independent sets of $M_1$ and
  $M_2$ of the form $X \cup Y$, where $Y\subseteq E'=E(M_1')=E(M_2')$ is a
  common independent set of $M_1'$ and $M_2'$. Note that \(E'=E\setminus (X \cup
  W)\). Algorithm $\mathcal{A}$ performs the following steps.

\begin{description}
\item[Step~1.] If $\omega(X)\geq w$, then output $X$ and return.

\item[Step~2.] Greedily compute at most $ks$ disjoint common independent sets
  $Y_1,\ldots,Y_\ell$ of $M_1'$ and $M_2'$, each of weight at least
  $w'=w-\omega(X)$, as follows.
  \begin{itemize}
  \item[(a)] Set \(i = 1, \mathit{Ys} = \{\}\).
  \item[(b)] If \(|\mathit{Ys}| = (i - 1) = ks\) then set \(\ell = (i - 1)\) and
    go to \textbf{Step~3}.
  \item[(c)] Set \(M_h''=M_h'- (\displaystyle\bigcup_{Y_{j} \in \mathit{Ys}}Y_{j})\), for $h=1,2$.
  \item[(d)] Find a common independent set $Z$ of $M_1''$ and $M_2''$ of the
    maximum weight.
  \item[(e)] If $\omega(Z)<w'$, then set \(\ell = (i - 1)\) and go to
    \textbf{Step~3}. Otherwise, set \(Y_{i} = Z, \mathit{Ys} = (\mathit{Ys} \cup
    \{Y_{i}\})\) and \(i = i + 1\), and go to \textbf{Step~2(b)}.
  \end{itemize}

\item[Step~3.] At this point we have \(\mathit{Ys} = \{Y_1,\ldots,Y_\ell\}\).
\begin{itemize}
\item If $\ell=0$ then return.
\item If $\ell=ks$ then output the sets $X\cup Y_1,\ldots,X\cup Y_\ell$ and return. 
\item If neither of the above holds then:
  \begin{itemize}
  \item Set \(R=\displaystyle\bigcup_{Y_{j} \in \mathit{Ys}}Y_{j}\).
  \item For each nonempty common independent set $Z\subseteq R$ of $M_1'$ and
    $M_2'$, set $W=R\setminus Z$ and recursively invoke $\mathcal{A}(X\cup
    Z,(M_1'-W)/Z,(M_2'-W)/Z)$.
\end{itemize}
\end{itemize}
\end{description}

This completes the description of $\mathcal{A}$. To construct $\mathcal{F}_w$,
we call $\mathcal{A}(\emptyset,M_1,M_2)$. Then the set $\mathcal{F}_w$ includes
all the sets output by $\mathcal{A}$. Note that in every recursive step we call
$\mathcal{A}(X\cup Z,(M_1'-W)/Z,(M_2'-W)/Z)$ only if $Z\neq\emptyset$. So the
size of the first argument $(X\cup Z)$ to a recursive call of \(\mathcal{A}\) is
strictly larger than the size of the first argument \(X\) of the parent call to
\(\mathcal{A}\). Moreover, since $Z$ is a common independent set of $M_1'$ and
$M_2'$, we have that $X\cup Z$ is a common independent set of $M_1$ and $M_2$.
Because the maximum size of the common independent set of $M_1$ and $M_2$ is at
most $s$, we obtain that the depth of the recursion is bounded by $s$, that is,
the algorithm is finite. We show the crucial property of $\mathcal{F}_w$
mentioned above.

\begin{claim}\label{cl:Fw}
  If $S=\{I_1,\ldots,I_k\}$ is a solution to the instance $(M_1,M_2,\omega,k,d)$
  such that $\omega(I_i)=w$ for some $i\in\{1,\ldots,k\}$, then there is an
  $I_i'\in \mathcal{F}_w$ such that \((S \setminus I_{i}) \cup I_i'\) is also a
  solution to $(M_1,M_2,\omega,k,d)$.
\end{claim}

\begin{claimproof}
  Fix a set \(I_{i} \in S\;;\;\omega(I_{i}) = w\). Recall that an arbitrary
  invocation of \(\mathcal{A}\) has the form \(\mathcal{A}(X, (M_{1} - W)/X,
  (M_{2} - W)/X)\) where \(X\) is a common independent set of \(M_{1}, M_{2}\)
  and \(W \subseteq (E \setminus X)\). For the very first invocation of
  \(\mathcal{A}\) these sets are \(X = \emptyset, W = \emptyset\), and these
  sets trivially satisfy the \textbf{viability condition} \((X \cup W) \cap
  I_{i} = X\); that is: \emph{\(I_{i}\) contains all of \(X\), and none of
    \(W\)}. We show that any invocation of \(\mathcal{A}\) whose arguments
  satisfy the viability condition \emph{either} outputs a set $I_i'$ that can be
  used to replace $I_i$ in $S$, \emph{or} makes at least one recursive call to
  \(\mathcal{A}\) such that the arguments to this recursive call satisfy the
  viability condition. Since the size of the first argument $(X\cup Z)$ to a
  recursive call of \(\mathcal{A}\) is strictly larger than the size of the
  first argument \(X\) of the parent call to \(\mathcal{A}\), we get that some
  call to \(\mathcal{A}\) will output a set \(I_{i}'\) with the desired
  property.
  
  Assume inductively that \(\mathcal{A}(X, M_{1}' = (M_{1} - W)/X, M_{2}' =
  (M_{2} - W)/X)\) is an invocation of \(\mathcal{A}\) whose arguments satisfy
  the viability condition. If $\omega(X)\geq w$, then $X=I_i$, because
  $\omega(I_i)=w$. In this case, the algorithm outputs $X=I_i$ in
  \textbf{Step~1}. Clearly, we can set $I_i'=I_i$, and we are done. So let us
  assume that this is not the case, and that $\omega(X)<w$. In this case \(X
  \subsetneq I_{i}\) holds and the algorithm goes to \textbf{Step~2}.

  Let \(E' = E \setminus (X \cup W)\) be the common ground set of \(M_{1}'\) and
  \(M_{2}'\). Let \(Y = (I_{i} \setminus X)\) and \(w'=\omega(Y)=w-\omega(X)\).
  Then since \(Y\) is a common independent set of \(M_{1}'\) and \(M_{2}'\), the
  greedy computation of \textbf{Step~2} produces a \emph{nonempty} family
  $Y_1,\ldots,Y_\ell$ of disjoint common independent sets of $M_1'$ and $M_2'$
  of weight at least $w'$ each. Note that $X\cup Y_1,\ldots,X\cup Y_\ell$ are
  common independent sets of $M_1$ and $M_2$. We consider two cases depending on
  the value of $\ell$ in \textbf{Step~3}. Note that by the above reasoning the
  case \(\ell = 0\) does not arise here.
  
  \begin{description}
  \item[Case~1.] $\ell=ks$. In this case the algorithm outputs the sets $X\cup
    Y_1,\ldots,X\cup Y_\ell$ where \((X \cap Y_{i}) = \emptyset\) holds for all
    \(i\in \{1,\dotsc,\ell\}\). Recall that every common independent set of
    $M_1$ and $M_2$ has size at most $s$. Therefore, the set
    $J=\bigcup_{j\in\{1,\ldots,k\},~j\neq i}I_j$ has size at most $(k-1)s<\ell$.
    Hence, by the pigeonhole principle, there is an $h\in\{1,\ldots,\ell\}$ such
    that $Y_h\cap I_j=\emptyset$ holds for all $j\in\{1,\ldots,k\}\;;\;j\neq i$.

    Let $I_i':=X\cup Y_h$. Then $I_i'$ is a common independent set of $M_1$ and
    $M_2$ and \(\omega(I_i') = \omega(X)+\omega(Y_h) \geq \omega(X)+w' =
    \omega(X)+\omega(I_{i} \setminus X) = \omega(I_{i})\), where the last
    equation follows from the fact that \(X\subseteq I_{i}\) holds. From this
    chain of relations we also get that \( \omega(I_{i} \setminus X) \leq
    \omega(Y_h)\) holds. Consider an arbitrary index $j\in
    \{1,\ldots,k\}\;;\;j\neq i$. Then \(\omega(I_{i} \sd I_{j}) = \omega(I_{j}
    \setminus I_{i}) + \omega(I_{i} \setminus I_{j}) \leq \omega(I_{j} \setminus
    X) + \omega(I_{i} \setminus I_{j}) = \omega(I_{j} \setminus X) + \omega(X
    \setminus I_{j}) + \omega((I_{i} \setminus X) \setminus I_{j}) \leq
    \omega(I_{j} \sd X) + \omega(I_{i} \setminus X) \leq \omega(I_{j} \sd X) +
    \omega(Y_{h})\). But since $I_i'=X\cup Y_h$, \((X \cap Y_{h}) = \emptyset\)
    and \((I_{j} \cap Y_{h}) = \emptyset\) we get that \(\omega(I_{j} \sd
    I_{i'}) = \omega(I_{j} \sd X) + \omega(Y_{h})\) holds. Thus \(\omega(I_{j}
    \sd I_{i'}) \geq \omega(I_{i} \sd I_{j})\), and so replacing $I_i$ by $I_i'$
    in the solution $S$ indeed gives us a solution to the instance
    $(M_1,M_2,\omega,k,d)$.
  \item[Case~2.]$0 < \ell < ks$. In this case we set
    \(R:=\bigcup_{i=1}^{\ell}Y_i\). From the construction we get that the
    matroids $M_1'-R$ and $M_2'-R$ have no common independent set of weight at
    least $w'$. Since \((I_{1} \setminus X) \setminus R\) is such a common
    independent set we have that \(\omega((I_{1} \setminus X) \setminus R) <
    w'\), and since \(\omega(I_{i}\setminus X)=w'\) we get that $Z=(I_i\setminus
    X)\cap R\neq\emptyset$. 
    Clearly, $Z$ is a common independent set of $M_1'$ and $M_2'$. Our algorithm
    considers all such sets. Hence there is a recursive call
    $\mathcal{A}(X',(M_1'-W)/Z,(M_2'-W)/Z)$ where \(Z = ((I_i\setminus X)\cap
    R), X' = X \cup Z, W=R\setminus Z\).
    By the choice of \(Z\) and \(W\) we get that \((X' \cup W) \cap I_{i} =
    X'\), so that this recursive call satisfies the viability
    condition. 
    Moreover, we have that $|X'|>|X|$. This completes the second case and the
    proof of the claim.\qedhere
  \end{description} 
\end{claimproof}

We already observed that the algorithm $\mathcal{A}$ is finite. Now we evaluate its running time and the size of $\mathcal{F}_w$.

\begin{claim}\label{cl:size-Fw}
The set $\mathcal{F}_w$ has size $2^{\Oh(s^2\log(ks))}$ and can be constructed in $2^{\Oh(s^2\log(ks))}\cdot |E|^{\Oh(1)}$ time. 
\end{claim}

\begin{claimproof}
  To give an upper bound on the size of $\mathcal{F}_w$, observe that in each
  recursive call, the algorithm $\mathcal{A}$ either outputs some sets, or
  performs some recursive calls, or simply returns without outputting anything.
  Notice that in Step~1, $\mathcal{A}$ can output at most one set, and
  $\mathcal{A}$ may output $ks$ sets in Step~3. The number of recursive calls is
  upper bounded by the number of nonempty common independent sets $Z\subseteq R$
  of $M_1'$ and $M_2'$. Since $\ell<ks$ and $|Y_i|\leq s$ for
  $i\in\{1,\ldots,\ell\}$, $|R|\leq ks^2$. Because for each $Z$, $|Z|\leq s$,
  the branching factor is at most $(ks^2)^s=2^{\Oh(s\log(ks))}$. Since the depth
  of the recursion is at most $s$, the search tree has $2^{\Oh(s^2\log(ks))}$
  leaves. This implies that the size of $\mathcal{F}_w$ is
  $2^{\Oh(s^2\log(ks))}$.

  To evaluate the running time, note that in Step~2, the algorithm greedily
  constructs the sets $Y_1,\ldots,Y_\ell$ that are common independent sets of
  $M_1'$ and $M_2'$. By \autoref{prop:intersection}, this can be done in
  polynomial time, because in each iteration we find a common independent set of
  maximum weight. Because the search tree has $2^{\Oh(s^2\log(ks))}$ leaves, the
  total running time is $2^{\Oh(s^2\log(ks))}\cdot |E|^{\Oh(1)}$.
\end{claimproof}

We construct $\mathcal{F}=\bigcup_{w=0}^{ds}\mathcal{F}_w$. By
\autoref{cl:size-Fw}, $|\mathcal{F}|\leq
(ds+1)\max_{w\in\{0,\ldots,ds\}}|\mathcal{F}_{w}|=2^{\Oh(s^2\log(ks))}\cdot d$
and $\mathcal{F}$ can be constructed in total $2^{\Oh(s^2\log(ks))}\cdot d\cdot
|E|^{\Oh(1)}$ time. \autoref{cl:Fw} implies that if $(M_1,M_2,\omega,k,d)$ is a
yes-instance of \probWDCIS, then the instance has a solution $I_1,\ldots,I_k$
with $I_i\in\mathcal{F}$ for $i\in\{1,\ldots,n\}$.
 \end{proof} 
 
 Combining \autoref{lem:big-is} and \autoref{lem:branching}, we obtain the main
 result of the section.
 
\divComIndSetsFPT*
  
\begin{proof} 
  Let $(M_1,M_2,\omega,k,d)$ be an instance of \probWDCIS. First, we use
  \autoref{prop:intersection} to solve \textsc{Matroid Intersection} for $M_1$
  and $M_2$ and find a common independent set $X$ of maximum size. If $|X|\geq
  k\lceil\frac{d}{2}\rceil$, then by \autoref{lem:big-is}, we conclude that
  $(M_1,M_2,\omega,k,d)$ is a yes-instance. Assume that this is not the case.
  Then the maximum size of a common independent set of $M_1$ and $M_2$ is $s<
  k\lceil\frac{d}{2}\rceil$. We apply \autoref{lem:branching} and construct the
  set $\mathcal{F}$ of size $2^{\Oh((kd)^2\log(kd))}$ in
  $2^{\Oh((kd)^2\log(kd))}\cdot |E|^{\Oh(1)} $ time. By this lemma, if
  $(M_1,M_2,\omega,k,d)$ is a yes-instance, it has a solution $I_1,\ldots,I_k$
  such that $I_i\in\mathcal{F}$ for $i\in \{1,\ldots,k\}$. Hence, to solve the
  problem we go over all $k$-tuples of the elements of $\mathcal{F}$, and for
  each $k$-tuple, we verify whether these common independent sets of $M_1$ and
  $M_2$ give a solution. Clearly, we have to consider $2^{\Oh(k^3d^2\log(kd))}$
  tuples. Hence, the total running time is $2^{\Oh(k^3d^2\log(kd))}\cdot
  |E|^{\Oh(1)}$.
\end{proof}  



\section{Perfect Matchings}\label{sec:matchings}

In this section we prove that \probPMD\ is fixed parameter tractable when
parameterized by \(k\) and $d$. We need the following simple observations later
in this section.

\begin{observation}
\label{obs:te}
The cardinality of symmetric differences of perfect matchings in a graph obeys
the triangle inequality. That is, for a graph $G$ and perfect matchings
$M_1,M_2,M_3$ in $G$, $\vert M_1\sd M_2\vert + \vert M_2\sd M_3\vert \geq \vert
M_1\sd M_3\vert$.
\end{observation}

Observation~\ref{obs:te} follows from the fact that Hamming distance is a metric and hence obeys triangular inequality.

\begin{observation}
\label{obs:setminussym}
Let $G$ be a graph and $M_1$ and $M_2$ be two perfect matchings in $G$. Then
$\vert M_1\sd M_2\vert=2\cdot \vert M_1\setminus M_2\vert=2\cdot \vert
M_2\setminus M_1\vert$.
\end{observation}

For an undirected graph $G$,  the Tutte matrix $\bfA$ of $G$ over the field ${\mathbb F}_2[X]$ is defined as follows,  where ${\mathbb F}_2$ is the Galois field on $\{0,1\}$ and $X=\{x_e ~\colon~ e\in E(G)\}$. The rows and columns of  $\bfA$ are labeled with $V(G)$ and for each $e=\{u,v\}\in E(G)$, $\bfA[u,v]=\bfA[v,u]=x_e$. All other entries in the matrix are zeros. That is, for any pair of vertices $u,v\in V(G)$, if there is no edge between $u$ and $v$, $\bfA[u,v]=0$. It is well known that $det(A)\neq 0$ if and only if $G$ has a perfect matching. As the characterstic 
of ${\mathbb F}_2$ is a $2$, the determinant of $\bfA$ coincides with the permanent of $\bfA$.  That is, 
\begin{equation}
\label{eqn:detper}
det(\bfA)=perm(\bfA)=\sum_{\sigma \in S_{V(G)}} \Pi_{v\in V(G)} \bfA[v,\sigma(v)]. 
\end{equation} 
Here,  $S_{V(G)}$ is the set of all permutations of $V(G)$. Let ${\sf PM}(G)$ be the set of perfect matchings on $G$. 
Then, one can show that    $det(\bfA)=\sum_{M\in {\sf PM}(G)} \Pi_{e\in M} x_e^2$.

Let $Y$ be a set of variables disjoint from $X$. For each edge $e$, let $L(e)\subseteq Y$ be a subset of variables. Let $\bfA'$ be the matrix obtained from $\bfA$ by replacing each entry of the form $x_e$ with $x_e\cdot \Pi_{y\in L(e)}y$. Then, 

\begin{eqnarray}
det(\bfA')=perm(\bfA')&=&\sum_{\sigma \in S_{V(G)}} \Pi_{v\in V(G)} \bfA'[v,\sigma(v)] \nonumber \\
&=&\sum_{\sigma \in S_{V(G)}} \Pi_{v\in V(G)} \left(\bfA[v,\sigma(v)] \cdot \Pi_{y\in L(\{v,\sigma(v)\})}y\right),  \label{eqn:detperlabel}
\end{eqnarray} 
where $\Pi_{y\in L(\{v,\sigma(v)\})}y=1$ if $\{v,\sigma(v)\}\notin E(G)$ or $L(e)=\emptyset$.  

\begin{lemma}
\label{lem:mon}
Let $G$ be an undirected graph and let $X=\{x_e ~\colon~e\in E(G)\}$ and $Y=\{y_1,\ldots,y_{\ell}\}$ be two  sets of variables such that $X\cap Y=\emptyset$.  For each edge $e\in E(G)$, we are also given a subset $L(e)\subseteq Y$. 
Let $\bfA'$ be the matrix defined as above.  
For any perfect matching $M$, $\Pi_{e\in M} x_e^2 \Pi_{y\in L(e)}y^2$ is a monomial in $det(\bfA')$. 
Moreover, for any monomial $m$ in $det(\bfA')$, $M'=\{e ~\colon~ x_e \mbox{ is a variable in }m\}$ is perfect a matching in $G$ and for each $e\in M'$, $L(e)$ is a subset of variables in the monomial $m$. 
\end{lemma}

\begin{proof}
A {\em cycle-matching cover} of $G$ is a subset of edges 
$F\subseteq E(G)$ such that $V(F)=V(G)$ and each connected component of $G[F]$ is either a cycle or an edge. 
Each non-zero term in the summation of \eqref{eqn:detperlabel}, there is a cycle-matching cover defined as follows. 
Let $\sigma\in S_{V(G)}$ such that $\Pi_{v\in V(G)} \bfA[v,\sigma(v)] \cdot \Pi_{y\in L(\{v,\sigma(v)\})}y$ is non-zero. 
Then,  $\Pi_{v\in V(G)} \bfA[v,\sigma(v)] $ is non-zero. As $G$ is a simple graph, $A[v,v]=0$. Therefore, since $\Pi_{v\in V(G)} \bfA[v,\sigma(v)] \neq 0$, there is no $1$-cycle in $\sigma$. Moreover any $\ell$-cycle in $\sigma$ corresponds to a cycle in $G$ and any $2$-cycle in $\sigma$ corresponds to an edge in $G$,  where the vertices covered in the cycle are the vertices present the cycle of the permutation. That is, for each cycle $(u_1,u_2,\ldots,u_{\ell}$ in $\sigma$, $u_1,u_2,\ldots,u_{\ell},u_1$ is a cycle in $G$ if $\ell>1$ and $u_1u_2$ is a matching edge 
if $\ell=2$.  Therefore, there is a cycle-matching cover corresponding to the non-zero term   $\Pi_{v\in V(G)} \bfA[v,\sigma(v)] \cdot \Pi_{y\in L(\{v,\sigma(v)\})}y$. 

Let $F$ be a cycle-matching cover. 
Let $\{C_1,\ldots, C_r\}$ be the set of cycles in $G[F]$ and $\{e_1,\ldots,e_{s}\}$ be the set of the edges in $F\setminus (\bigcup_{i} E(C_i))$. Let $F'=\bigcup_{i\in [r]} E(C_i)$. For each cycle $C=u_1,u_2,\ldots,u_{\ell},u_1$ in $G[F]$, where $\ell>2$ one can define two permutations $\sigma_1$ and $\sigma_2$ on $V(C)$ as follows: $(u_1,u_2,\ldots,u_{\ell})$ and $(u_1,u_{\ell},u_{\ell-1},\ldots,u_{2})$. That is,  
\[\Pi_{v\in V(C)} \bfA[v,\sigma_1(v)]=\Pi_{v\in V(C)} \bfA[v,\sigma_2(v)]=\Pi_{e\in E(C)}\left(x_e\cdot \Pi_{y\in L(e)}y\right).\]

This implies that there are $2^{r}$ terms in \eqref{eqn:detperlabel} which are equal to $\sum_{e\in F\setminus F'}\left(x_e\cdot \Pi_{y\in L(e)}y\right)^2 + \sum_{e\in F'}\Pi_{e\in F'}\left(x_e\cdot \Pi_{y\in L(e)}y\right)$ (which we call the terms corresponding to $F$).
In other words if $F$ is a perfect matching then 
$\Pi_{e\in F}\left(x_e\cdot \Pi_{y\in L(e)}y\right)^2$ is a unique term in  \eqref{eqn:detperlabel}, and if $F$ is has 
cycle  then the terms corresponding to $F$ will cancel each other, because the characteristic of the ${\mathbb F}_2[X\cup Y]$ is~$2$. Therefore, for any perfect matching $M$, $\Pi_{e\in M} x_e^2 \Pi_{y\in L(e)}y^2$ is a monomial in $det(\bfA')$.

As any non-zero term in \eqref{eqn:detperlabel} corresponds to a cycle-matching cover, and for any cycle matching cover that contains at least one cycle all the terms corresponding to it cancels each other we have the following. For any monomial $m$ in $det(\bfA')$, $M'=\{e ~\colon~ x_e \mbox{ is a variable in }m\}$ is perfect a matching in $G$.  Also, from the construction of $\bfA'$, it follows that for each $e\in M'$, $L(e)$ is subset of variables in $m$.  
\end{proof}

We use the following two known results. 

\begin{proposition}[Schwartz-Zippel Lemma~\cite{Schwartz,Zippel}]
\label{prop:SZ}
Let $P(x_1, \ldots, x_n)$ be a multivariate polynomial of total degree at most $d$ over a field ${\mathbb F}$, and $P$ is not identically zero. Let $r_1,\ldots,r_n$ be the elements in ${\mathbb F}$ choses uniformly at random with repetition. Then $\prob(P(r_1, . . . , r_n) = 0) \leq \frac{d}{|{\mathbb F}|}$.
\end{proposition}

For a multivariate polynomial $P$ and a monomial $m$, we let $P(m)$ denote the coefficient of $m$ in $P$.

\begin{proposition}[\cite{Wahlstrom13}]
\label{prop:magnus}
Let $P(x_1, \ldots , x_n)$ be a polynomial over a field of characteristic two, and $T \subseteq  [n]$ 
be a set of target indices. For a set $I \subseteq  [n]$, define $P_{-I} (x_1, \ldots , x_n) = P(y_1, \ldots , y_n)$ where $y_i = 0$  for $i \in I$ and $y_i = x_i$ otherwise. Define
$$Q(x_1,\ldots, x_n) = \sum_{I\subseteq T} P_{-I} (x_1, \ldots , x_n).$$
Then, for any monomial m such that $t := \Pi_{i\in T}x_i$ divides $m$ we have $Q(m) = P(m)$, and for
every other monomial we have $Q(m) = 0$.
\end{proposition}

\begin{lemma}
\label{lem:PMDstep1}
There is an algorithm that given an undirected graph $G$, perfect matchings $M_1,\ldots, M_r$, and a non-negative integer $s$, runs in time $2^{\Oh(rs)} n^{\Oh(1)}$, and outputs a perfect matching $M$ such that $\vert M \setminus M_i\vert \geq s$ for all $i\in \{1,\ldots,r\}$ (if such a matching exists) with probability at least $\frac{2}{3}e^{-rs}$.   
\end{lemma}

\begin{proof}
For each $i\in \{1,\ldots,r\}$, we color each edge in $E(G)\setminus M_i$ uniformly at random using colors $\{c_{i,1},\ldots,c_{i,s}\}$. Now we label each edge with a subset of the variable set $Y=\{y_{i,j}~\colon~ i\in [r], j\in [s]\}$.   For each edge $e$, we label it with $$L(e)=\{y_{i,j}~\colon~ i\in [r]  \mbox{ and $e$ is colored with $c_{i,j}$ in the random coloring for $i$}\}.$$ Let $\bfA$ be the 
Tutte matrix of $G$ over the field ${\mathbb F}_2[X]$, where $X=\{x_e \colon e\in E(G)\}$. 
Let $\bfA'$ be the matrix obtained from $\bfA$ by replacing each entry of the form $x_e$ with $x_e\cdot \Pi_{y\in L(e)}y$. 

Suppose there is a matching $M$ such that $\vert M \setminus M_i\vert \geq s$ for all $i\in \{1,\ldots,r\}$. Then, for each  $i\in [r]$, let $\{e_{i,1},\ldots,e_{i,s}\}\subseteq (M \setminus M_i)$ be an arbitrary subset.  We say that 
$\{e_{i,1},\ldots,e_{i,s}\}$ is colorful if the edges in $\{e_{i,1},\ldots,e_{i,s}\}$ gets distinct colors from $\{c_{i,1},\ldots,c_{i,s}\}$ in the random coloring for $i$. 
Then for each $i\in [r]$, 
the probability that $\{e_{i,1},\ldots,e_{i,s}\}$
is colorful is $\frac{s!}{s^s}\geq e^{-s}$. 
For each $q\in [r]$, let $E_q$ be the even that $\{e_{q,1},\ldots,e_{q,s}\}$ is colorful. 
As the random coloring for $i\in  [r]$ is different from the random coloring for $j\in [r]\setminus \{i\}$, the events 
$E_i$ and $E_j$ are independent. That is, $E_1,\ldots,E_r$ are independent events and hence $\prob[\bigcap_{i=1}^r E_i]\geq e^{-rs}$. Therefore,  there is a monomial $m$ in $det(\bfA')$ with probability at least 
$e^{-rs}$ such that $M=\{e\in E(G) ~\colon~ x_e \mbox{ is a variable in }m\}$ and $Y$ is a subset of variables in $m$. 

Now, suppose there is a monomial $m$ in $det(\bfA')$ such that $Y$ is a subset of variables in $m$. 
Therefore, since  for each $i\in [r]$ only the edges in $E(G)\setminus M_i$ are colored and $\{y_{i,j}~\colon~j\in [s]\}\subseteq Y$, we have that $\vert M\setminus M_i\vert \geq s$. Moreover, by Lemma~\ref{lem:mon}, $\{e \in E(G)~\colon~ x_e \mbox{ is a variable in }m\}$ is a perfect matching in $G$. 

Thus, it is enough to check whether there exists a monomial $m$ in $det(\bfA')$ such that all the variables in $Y$ are present in $m$.

\begin{claim}
\label{claim:matchingoutput}
There is an algorithm, that runs in time $2^{\Oh(|Y|)} n^{\Oh(1)}$ and it outputs the following. 
If there is no monomial in $det(\bfA')$ that contains $Y$, then the algorithm outputs No. If there is a monomial in $det(\bfA')$ that contains $Y$, then the algorithm outputs $Z\subseteq X$ with probability at least $2/3$ such that there is a monomial $m$ in $det(\bfA')$ with variables in $m$ is exactly equal to $Z\cup Y$. 
\end{claim}

\begin{claimproof}
Let $\{e_1,\ldots,e_{m}\}=E(G)$.  Let $P(x_{e_1},\ldots,x_{e_m},y_{1,1},\ldots,y_{r,s})=det(\bfA')$. For each  $I\subseteq Y$, we define $P_{-I}(x_{e_1},\ldots,x_{e_m},y_{1,1},\ldots,y_{r,s})=P(x_{e_1},\ldots,x_{e_m},z_{1,1},\ldots,z_{r,s})$, where  for all $i\in [r]$ and $j\in [s]$, $z_{i,j}=0$ if $y_{i,j}\in I$ and $z_{i,j}=y_{i,j}$ otherwise. 
Let 
$$Q=Q(x_{e_1},\ldots,x_{e_m},y_{1,1},\ldots,y_{r,s})=\sum_{I\subseteq Y} P_{-I}(x_{e_1},\ldots,x_{e_m},y_{1,1},\ldots,y_{r,s}).$$

By Proposition~\ref{prop:magnus}, the set of monomials in $Q$ are the set of monomial is $det(\bfA')$ that contains $Y$. Our objective is to find out the variables in such a monomial if it exists. Towards that we consider $Q$ be a polynomial in a field extension ${\mathbb F}'$ of ${\mathbb F_2}$ such that the number of elements in   the field ${\mathbb F}'$ is at least $t$, which we fix later.  From the construction of $det(\bfA')$, we know that the degree 
of $det(\bfA')$ and $Q$ is at most $d=n+2rs$. By Proposition~\ref{prop:SZ}, the existence of a required monomial can be tested in polynomial time with failure probability at most $\frac{d}{t}$. But, recall that we need to find out the variables in such a monomial. Towards that we do the following.  Notice that if $Q$ is a polynomial identically zero, then our answer is No.  This can be checked using Proposition~\ref{prop:SZ}. Assume that $Q$ is not identically zero. Let $Q_1=Q(0,x_{e_2},\ldots,x_{e_m},y_{1,1},\ldots,y_{r,s})$. 
If $Q_1$ is not identically equal to zero, then there is a monomial satisfying the property mentioned in the statement of the claim is present in $Q_1$. Otherwise we know that every monomial with the required property contains the variable $x_{e_1}$. Again this can be checked using Proposition~\ref{prop:SZ}. So if $Q_1$ is not identically zero, then set $Q=Q_1$. Next, we let $Q_2=Q(x_{e_1},0,x_{e_2},\ldots,x_{e_m},y_{1,1},\ldots,y_{r,s})$. 
Again, $Q_2$ is not identically equal to zero, then there is a monomial  satisfying the property mentioned in the statement of the claim, is present in $Q_2$. Otherwise we know that every monomial with the required property contains the variable $x_{e_2}$. By repeating this process at most $m$ times, we will be able to obtain all the variables present in the required monomial. Our algorithm will succeed if  all the $m+1$ application of Proposition~\ref{prop:SZ} do not fail. Thus, by union bound the failure probability of  is at most $(m+1)\frac{d}{t}$. We set $(m+1)\frac{d}{t}=\frac{1}{3}$ and this implies that $t=(m+1)(n+rs)$. Hence the success probability of our algorithm is at least $\frac{2}{3}$.  

Towards the running time analysis, notice that the construction of the polynomial $Q$ takes time $2^{\vert Y\vert} n^{\Oh(1)}$ and each application of Proposition~\ref{prop:SZ} takes time polynomial in the size of $Q$. This implies that the total running time is bounded by $2^{\Oh(|Y|)} n^{\Oh(1)}$. 
\end{claimproof}

Now we run the algorithm in Claim~\ref{claim:matchingoutput}, and get a subset $Z\subseteq \{x_e~\colon~ e\in E(G)\}$ (if it exists) such that there is a monomial  $m$ in $det(\bfA')$ and the variables in $m$ is exactly equal to 
$Z\cup Y$ with probability at least $2/3$. As the initial random coloring of edges succeeds with probability at least 
$e^{-rs}$, the success probability of out algorithm is at least $\frac{2}{3}e^{-rs}$. If no such monomial exists then the algorithm outputs No. By Lemma~\ref{lem:mon}, $M=\{e~\colon~ x_e\in Z\}$ is a perfect matching and it is the required output. The running time of the algorithm follows from Claim~\ref{claim:matchingoutput}. 
This completes the proof of the lemma. 
\end{proof}

\begin{lemma}
\label{lem:PMDstep2}
There is an algorithm that given an undirected graph $G$, a perfect matching $M$, and  non-negative integers $r,d,s$, runs in time $2^{\Oh(r^2s)} n^{\Oh(1)}$, and outputs $r$ 
perfect matchings $M_1^{\star},\ldots,M_r^{\star}$ such that $\vert M \sd M_i^{\star}\vert \leq s$ for all $i\in \{1,\ldots,r\}$ 
and $\vert M_i^{\star} \sd M_j^{\star} \vert\geq d$ for all distinct $i,j\in [r]$ (if such  matchings exist) with probability at least $e^{-rs}$.   If no such perfect matchings exist, then the algorithm outputs No 
\end{lemma}

\begin{proof}
Suppose there exist perfect matchings $M_1,\ldots,M_r$ such that $\vert M \sd M_i\vert \leq s$ for all $i\in \{1,\ldots,r\}$ and $\vert M_i \sd M_j \vert\geq d$ for all distinct $i,j\in [r]$. Then we know that $\sum_{i=1}^r\vert M \sd M_i\vert \leq rs$. Let $S_i=M \sd M_i$ for all $i\in [r]$. Notice that, as $M$ and $M_i$ are perfect matchings $S_i$ forms a collection of alternating cycles (i.e., edges in the cycles alternate between $M$ and $M_i$). 
Let $S=\bigcup_{i=1}^r S_i$. 

We do a random coloring on the edges of $G$ using $rs$ colors. That is, we color each edge of $G$ uniformly at random with a color from $\{1,\ldots,rs\}$. We say that the random coloring is {\em good} if all the edges 
in   $S$ gets distinct colors. The probability that the random coloring is good is $e^{-rs}$. 

 Now on assume that the random coloring is good. 
For each $i\in [r]$, let $C_i$ be the set of colors on the edges $S_i$. Notice that $\vert C_i\vert =\vert S_i \vert$. 
\begin{claim}
For any two distinct integers $i,j\in [r]$, $\vert C_i\sd C_j \vert \geq d$. 
\end{claim}
\begin{claimproof}
We know that $\vert M_i \sd M_j \vert \geq d$. Let $E_i=M_i\setminus M_j$ and $E_j=M_j\setminus M_i$. 
Notice that $\vert E_i\vert + \vert E_j\vert \geq d$ and $E_i\cap E_j=\emptyset$. 

Let $E_{i,1}=E_i\cap M$, $E_{i,2}=E_i\setminus E_{j,1}$, $E_{j,1}=E_j\cap M$, and $E_{j,2}=E_j\setminus E_{j,1}$. 
As $E_{i,1} \subseteq M\setminus M_j$ and $E_{i,1}\subseteq M\cap M_i$, we have that $E_{i,1}\subseteq S_{j}\setminus S_i$. Similarly $E_{j,1}\subseteq S_i\setminus S_j$.  As $E_{i,2}\subseteq M_i \setminus (M_j\cup M)$, we have that $E_{i,2}\subseteq S_i\setminus S_j$. Similarly, we have $E_{j,2}\subseteq S_j\setminus S_i$. That is, 
we prove that $E_{i,1}\cup E_{j,2} \subseteq S_{j}\setminus S_i$ and $E_{j,1}\cup E_{i,2} \subseteq S_{i}\setminus S_j$. Also, since all the edges in $S$ gets distinct colors and the colors on the edges in $S_i$ and $S_j$ are $C_i$ and $C_j$, respectively,  we have that  $\vert C_i\sd C_j \vert \geq \vert E_i\cup E_j\vert \geq d$.  
This completes the proof of the claim. 
\end{claimproof}

Next we prove the reverse direction of the above claim. 

\begin{claim}
  Let $Q_1$ and $Q_2$ be two collections of alternating cycles in $G$ (i.e., the
  edges in $Q_i$ are alternating between $M$ and $E(G)\setminus M$ for each
  $i\in \{1,2\}$) such that following hold: $\vert E(Q_1)\vert =\vert C_1\vert$,
  $\vert E(Q_2)\vert =\vert C_2\vert$, the edges in $E(Q_1)$ uses distinct
  colors from $C_1$, and the edges in $E(Q_2)$ uses distinct colors from $C_2$.
  Let $P_1= E(Q_1)\sd M$ and $P_2= E(Q_2)\sd M$. Then, $P_1$ and $P_2$ are
  perfect matchings and $\vert P_1\sd P_2 \vert \geq d$.
\end{claim}

\begin{claimproof}
  As $M$ is a perfect matching and $Q_1$ is a collection of alternating cycles,
  we have that $P_1=E(Q_1)\sd M$ is a perfect matching. By similar arguments, we
  have that $P_2$ is a perfect matching. Now we prove that $E(Q_1)\sd
  E(Q_2)\subseteq P_1 \sd P_2$. Consider an edge $e\in E(Q_1)\setminus E(Q_2)$.
  We have two cases based on whether $e\in M$ or not. In the first case, assume
  that $e\in M$. Since $e\in E(Q_1)$, $e\in M$, and $P_1=E(Q_1)\sd M$, we have
  that $e\notin P_1$. Also since $e\notin E(Q_2)$, $e\in M$, and $P_2=E(Q_2)\sd
  M$, we have that $e\in P_2$. Therefore, $e\in P_1 \sd P_2$.

  For the second case, we have that $e\notin M$. Since $e\in E(Q_1)$, $e\notin
  M$, and $P_1=E(Q_1)\sd M$, we have that $e\in P_1$. Also, since $e\notin
  E(Q_2)$ and $e\notin M$, we have that $e\notin P_2$. Therefore, $e\in P_1\sd
  P_2$.

  By arguments, similar to above, one can prove that an edge $e'\in
  E(Q_2)\setminus E(Q_1)$ also belongs to $P_1\sd P_2$. Thus, we proved that
  $E(Q_1)\sd E(Q_2)\subseteq P_1 \sd P_2$. Since $\vert E(Q_1)\vert =\vert
  C_1\vert$, $\vert E(Q_2)\vert =\vert C_2\vert$, the edges in $E(Q_1)$ uses
  distinct colors from $C_1$, and the edges in $E(Q_2)$ uses distinct colors
  from $C_2$, we have that $\vert E(Q_1)\sd E(Q_2)\vert \geq \vert C_1\sd
  C_2\vert\geq d$. Therefore, $\vert P_1 \sd P_2\vert \geq \vert E(Q_1)\sd
  E(Q_2)\vert \geq d$.
\end{claimproof}

Thus, to prove the lemma, it is enough to find a collection $Q_i$ of alternating cycles such that  $\vert E(Q_i)\vert =\vert C_i\vert$,  the edges in $E(Q_i)$ uses distinct colors from $C_i$, for each $i\in [r]$. That is, our algorithm guesses $C_1,\ldots C_r$ and computes $Q_1,\ldots,Q_r$. The cost of guessing $C_1,\ldots,C_r$ is $2^{r^2s}$.  
Now, given $C_i$, to compute $Q_i$ with desired property ($\vert E(Q_i)\vert =\vert C_i\vert$,  the edges in $E(Q_i)$ are colored with distinct colors from $C_i$), we design a simple dynamic programming (DP) algorithm. We give a brief outline of this algorithm below. For each subset $L\subseteq C_i$ and pair of vertices $u,v$ we have table entries $D[L,u,v]$ and $D[L,\bot,\bot]$ which stores the following. If there is a collection $Q$ of alternating cycles and an alternating path with $u$ and $v$ as endpoints such that $\vert E(Q)\vert =\vert L\vert$ and the edges in $E(Q)$ are colored with distinct colors from $L$, then we store one such collection in $D[L,u,v]$. Otherwise, we store $\bot$ in $D[L,u,v]$. If there is a collection $Q'$ of alternating cycles such that $\vert E(Q')\vert =\vert L\vert$ and the edges in $E(Q')$ are colored with distinct colors from $L$, then we store one such collection in $D[L,\bot,\bot]$. Otherwise, we store $\bot$ in $D[L,\bot,\bot]$.  We compute the DP table entries in the increasing order of the size of $L$. The base case is when $L=\emptyset$. That is, $D[\emptyset,\bot,\bot]=\emptyset$ and 
$D[\emptyset,u,v]=\bot$ for any two vertices $u$ and $v$. Now, for any $\emptyset \neq L\subseteq C_i$, and two distinct vertices $u,v\in V(G)$, we compute $D[L,u,v]$ and $D[L,\bot,\bot]$ as follows. If there is an edge $(x,y)\in E(G)$ such that the color $c$ of $(x,y)$ belongs to $L$, and $D[L\setminus \{c\},x,y]=Q_1\neq \bot$, then we store the graph induced on $E(Q_1)\cup \{(x,y)\}$ in $D[L,\bot,\bot]$. Otherwise we store $\bot$ in  $D[L,\bot,\bot]$. Also, if there a vertex $w$ adjacent to $v$ such that 
the color $c$ of $(w,v)$ belongs to $L$, and $D[L\setminus \{c\},u,w]=Q_1'\neq \bot$, then we store the graph induced on $E(Q'_1)\cup \{(w,v)\}$ in $D[L,u,v]$. Otherwise we store $\bot$ in  $D[L,u,v]$. At the end we output $D[C_i,\bot,\bot]$. 

By using standard induction, one can prove that the computation of $D[C_i,\bot,\bot]$ is correct and the details is omitted here. 
As the number of table entries for $D[.,.,.]$ is upper bounded by $2^{rs+1}n^2$, the running time to compute 
$Q_i$ is $2^{rs} n^{\Oh(1)}$. We have already mentioned that the cost of guessing $C_1,\ldots,C_r$ is $2^{r^2s}$. Therefore, the total running time to compute the required $r$ perfect matchings is $2^{\Oh(r^2s)} n^{\Oh(1)}$.
\end{proof}

Finally, we put together both the lemmas and prove the main theorem of the section. 

\divMatchingFPT*

\begin{proof}
  Let $(G,k,d)$ be the input instance. Our algorithm ${\cal A}$ has two steps.
  In the first step of ${\cal A}$ we compute a collection of matchings greedily
  such that they are far apart using Lemma~\ref{lem:PMDstep1}. Towards that
  first we run an algorithm to compute a maximum matching in $G$ and let $M_1$
  be the output. If $M_1$ is not a perfect matching we output No and stop. Next
  we iteratively apply Lemma~\ref{lem:PMDstep1} to compute a collection of
  perfect matchings that are far apart. Formally, at the beginning of step $i$,
  where $\leq 1\leq i<k$, we have perfect matchings $M_1,\ldots,M_i$ such that
  $\vert M_{j}\setminus M_{j'}\vert \geq 2^{k-i}d$ for any two distinct $j,j'\in
  \{1,\ldots,i\}$. Now, we apply Lemma~\ref{lem:PMDstep1} with $r=i$ and
  $s=2^{k-i-1}d$ and it will either output a matching $M_{i+1}$ such that $\vert
  M_{i+1}\setminus M_j\vert \geq 2^{k-i-1}d$ for all $j\in \{1,\ldots,i\}$, or
  not. If no such matching exists, then the first step of the algorithm ${\cal
    A}$ is complete. So at the end of the first step of the algorithm ${\cal
    A}$, we have perfect matchings $M_1,\ldots,M_q$, where $q\in\{1,\ldots,k\}$
  such that
\begin{romanenumerate}
\item for any two distinct integers $i,j\in \{1,\ldots,q\}$, $|M_i\setminus M_j|\geq 2^{k-q}d$, and 
\item if $q\neq k$, then for any other perfect matching $M\notin
  \{M_1,\ldots,M_q\}$, $\vert M\setminus M_j\vert \leq 2^{k-q-1}d$.
\end{romanenumerate}

If $q= k$, then $\{M_1,\ldots,M_k\}$ is a solution to the instance $(G,k,d)$, and hence our algorithm ${\cal A}$ outputs Yes. Now on, we assume that $q\in \{1,\ldots,k-1\}$. Statements $(i)$ and $(ii)$, and 
Observation~\ref{obs:setminussym} imply that  

\begin{romanenumerate}\addtocounter{enumi}{2}
\item for any two distinct integers $i,j\in \{1,\ldots,q\}$, $|M_i\sd M_j|\geq 2^{k-q+1}d$, and 
\item for any perfect matching $M\notin \{M_1,\ldots,M_q\}$, $\vert M\sd M_j\vert < 2^{k-q}d$.  
\end{romanenumerate}

Statements $(ii)$ and $(iv)$, and Observation~\ref{obs:te} imply the following claim.
\begin{claim}
\label{clm:uniquemat}
For any perfect matching $M$, there exists a unique $i\in \{1,\ldots,q\}$ such that $\vert M\sd M_i\vert < 2^{k-q}d$.
\end{claim}

Let ${\cal M}=\{M_1^{\star},\ldots,M^{\star}_k\}$ is a solution to the instance $(G,k,d)$. Then, by Claim~\ref{clm:uniquemat}, there is a partition of ${\cal M}$ into ${\cal M}_1\uplus \ldots \uplus {\cal M}_q$ (with some blocks possibly being empty) such that for each $i\in \{1,\ldots,q\}$, and each $M\in {\cal M}_i$, $\vert M\sd M_i\vert \leq 2^{k-q}d$. Thus, in the second step of our algorithm ${\cal A}$, we guess $r_1=\vert {\cal M}_1\vert,\ldots, r_q=\vert {\cal M}_q\vert$ and apply Lemma~\ref{lem:PMDstep2}. That is, for each $i\in \{1,\ldots,q\}$ such that $r_i\neq 0$, we apply Lemma~\ref{lem:PMDstep2} with $M=M_i$, $r=r_i$, and $s=2^{k-q}d$. Then for each $i\in {1,\ldots,q}$, let the output of Lemma~\ref{lem:PMDstep2} be $N_{i,1},\ldots,N_{r_i}$. Clearly $\vert N_{i,j}\sd N_{i,j'}\vert \geq d$ for any two distinct $j,j'\in \{1,\ldots,r_i\}$. Observation~\ref{obs:te} and statement $(iii)$ implies that for any two distinct $i,j\in \{1,\ldots,q\}$, 
the cardinality of the symmetric difference between a matching in $\{N_{i,1},\ldots,N_{i,r_i}\}$ and a matching in $\{N_{j,1,\ldots,N_{j,r_j}}\}$  is at least $d$. 

If algorithm ${\cal A}$ computes a solution in any of the guesses for $r_1,\ldots,r_d$, then we output Yes. Otherwise we output No. As the number of choices for $r_1,\ldots r_k$ is upper bounded by $k^{\Oh(k)}$, from Lemmas~\ref{lem:PMDstep1} and \ref{lem:PMDstep2} we get that the running time of ${\cal A}$ is $2^{2^{\Oh(kd)}}n^{\Oh(1)}$ and the success probability is at least $2^{-2^{ckd}}$ for some constant $c$. To get success probability 
$1-1/e$, we do $2^{2^{ckd}}$ many executions  of ${\cal A}$  and output Yes if we succeed in at least one of the iterations and output No otherwise. Thus, running time of the overall algorithm is  $2^{2^{\Oh(kd)}}n^{\Oh(1)}$. 
\end{proof}


\section{Conclusion}\label{sec:conclusion}
We took up weighted diverse variants of two classical matroid problems and the
unweighted diverse variant of a classical graph problem. We showed that the two
diverse matroid problems are \classNP-hard, and that the diverse graph problem
cannot be solved in polynomial time even for the smallest sensible measure of
diversity. We then showed that all three problems are \classFPT with the
combined parameter \((k,d)\) where \(k\) is the number of solutions and \(d\) is
the diversity measure.

We conclude with a list of open questions:
\begin{itemize}
\item We showed that the unweighted, counting variant of \probWDB does not have
  a polynomial-time algorithm unless \(\classP=\classNP\)
  (\autoref{thm:WDB_NP_hard}). This is the case when all the weights are \(1\)
  and \(d=1\) or \(d=2\). Both the weighted and unweighted variants can be
  solved in polynomial time when \(k=1\) (the greedy algorithm) and \(k=2\)
  ((weighted) matroid intersection). What happens for larger, constant values of
  \(d\) and/or \(k\)? Till what values of \(d,k\) does the problem remain
  solvable in polynomial time? These questions are interesting also for special
  types of matroids. For instance, is there a polynomial-time algorithm that
  checks if an input graph has \emph{three} spanning trees whose edge sets have
  pairwise symmetric difference at least \(d\), or is this already
  \classNP-hard?
\item A potentially easier question along the same vein would be: we know from
  \autoref{thm:WDB_NP_hard} that \probWDB is unlikely to have an \classFPT
  algorithm parameterized by \(d\) alone. Is \probWDB \classFPT parameterized by
  \(k\) alone?
\item Unlike for the other two problems, we don't have hardness results for
  \probWDCIS for small values of \(k\) or \(d\). Is \probWDCIS \classFPT when
  parameterized by either \(d\) or \(k\)? Is this problem in \classP when all
  the weights are \(1\)?
\end{itemize}



\bibliography{matroids} 

\end{document}